\documentclass[a4paper, onecolumn, draftcls, journal]{IEEEtran}

\usepackage{subcaption}
\usepackage{cancel}
\usepackage{graphicx}
\usepackage[dvipsnames]{xcolor}
\usepackage{epsfig}
\usepackage{epstopdf}
\usepackage{mathtools}
\usepackage{etoolbox}
\usepackage{wrapfig}
\usepackage{cite}
\usepackage{enumerate}
\usepackage{amsmath,amsthm}
\usepackage{amssymb}

\newcommand{\Cov}{\mathrm{cov}}

\DeclareMathOperator{\trace}{tr}
\DeclareMathOperator{\diag}{diag}

\newtheorem{theorem}{Theorem}
\newtheorem{lemma}{Lemma}

\newtheorem{corollary}{Corollary}

\newtheorem{remark}{Remark}
\newtheorem{problem}{Problem}
\newtheorem{assumption}{Assumption}

\newtheorem{proposition}{Proposition}

\newcommand{\T}{^{\mbox{\tiny T}}}

\newcommand{\sr}{\stackrel}

\newcommand{\rar}{\rightarrow}

\newcommand{\tri}{\sr{\triangle}{=}}

\newcommand{\be}{\begin{equation}}
\newcommand{\ee}{\end{equation}}
\newcommand{\bea}{\begin{eqnarray}}
\newcommand{\eea}{\end{eqnarray}}
\newcommand{\bes}{\begin{eqnarray*}}
\newcommand{\ees}{\end{eqnarray*}}
\newcommand{\beae}{\begin{IEEEeqnarray}{rCl}}
\newcommand{\eeae}{\end{IEEEeqnarray}}

\newcommand{\bi}{\begin{itemize}}
\newcommand{\ei}{\end{itemize}}
\newcommand{\ben}{\begin{enumerate}}
\newcommand{\een}{\end{enumerate}}
\newcommand{\bp}{\begin{problem}}
\newcommand{\ep}{\end{problem}}
\newcommand{\hso}{\hspace{.1in}}
\newcommand{\hst}{\hspace{.2in}}

\newcommand{\noi}{\noindent}

\begin{document}
\title{Structural Properties of Optimal Test Channels for Distributed Source Coding with Decoder Side Information  for  Multivariate Gaussian Sources with Square-Error Fidelity}

\author{
\IEEEauthorblockN{
  Michail Gkagkos\IEEEauthorrefmark{1}
   and
     Charalambos D. Charalambous\IEEEauthorrefmark{2}\\
   \IEEEauthorblockA{
    \IEEEauthorrefmark {1}Department of Electrical and Computer Engineering\\
     Texas A\&M University, College Station, Texas \\
    Email: gkagkos@tamu.edu}\\
   \IEEEauthorblockA{
     \IEEEauthorrefmark{2}Department of Electrical and Computer Engineering,\\
      University of Cyprus, Nicosia, Cyprus\\
    Email:  chadcha@ucy.ac.cy}
}}

\maketitle

\begin{abstract}
This paper focuses on the structural properties of test channels, of   Wyner's \cite{wyner1978}   operational information rate distortion function (RDF),  $\overline{R}(\Delta_X)$,  of   a tuple of multivariate correlated,  jointly independent and identically distributed  Gaussian random variables (RVs),  $\{X_t, Y_t\}_{t=1}^\infty$, $X_t: \Omega  \rar {\mathbb R}^{n_x}$, $Y_t: \Omega \rar {\mathbb R}^{n_y}$, with average mean-square error at the decoder, $\frac{1}{n} {\bf E}\sum_{t=1}^n||X_t - \widehat{X}_t||^2\leq \Delta_X$,    when $\{Y_t\}_{t=1}^\infty$ is the side information available to the decoder only.   We construct optimal test channel realizations, which achieve the informational  RDF, $\overline{R}(\Delta_X) \tri\inf_{{\cal M}(\Delta_X)} I(X;Z|Y)$, where ${\cal M}(\Delta_X)$  is the set of auxiliary RVs $Z$  such that,  ${\bf P}_{Z|X,Y}={\bf P}_{Z|X}$,  $\widehat{X}=f(Y,Z)$, and ${\bf E}\{||X-\widehat{X}||^2\}\leq \Delta_X$. We show the fundamental structural properties: (1) Optimal test channel realizations that achieve the RDF, $\overline{R}(\Delta_X)$,  satisfy conditional independence, 
\begin{align}
 {\bf P}_{X|\widehat{X}, Y, Z}={\bf P}_{X|\widehat{X},Y}={\bf P}_{X|\widehat{X}}, \hst  {\bf E}\Big\{X\Big|\widehat{X}, Y, Z\Big\}={\bf E}\Big\{X\Big|\widehat{X}\Big\}=\widehat{X}
\end{align} 
and     (2)  similarly for the conditional RDF, ${R}_{X|Y}(\Delta_X) \tri \inf_{{\bf P}_{\widehat{X}|X,Y}:{\bf E}\{||X-\widehat{X}||^2\} \leq \Delta_X}  I(X; \widehat{X}|Y)$, when $\{Y_t\}_{t=1}^\infty$ is  available to both the encoder and  decoder, and the equality $\overline{R}(\Delta_X)={R}_{X|Y}(\Delta_X)$.

This paper also shows  that  the optimal test channel realization of  the RDF of  
the distributed remote source coding problem,  
\cite[Theorem~4]{tian-chen2009} and \cite[Theorem~3A]{Zahedi-Ostegraard-2014}  (a  noisy version of Wyner's $\overline{R}(\Delta_X)$), when specialized to Wyner's RDF $\overline{R}(\Delta_X)$,  do not generate  Wyner's value of  this   RDF of scalar RVs. 

\end{abstract}

\section{Introduction, Problem Statement and Main Results}
\subsection{Wyner \cite{wyner1978} and Wyner and Ziv \cite{wyner-ziv1976} Lossy Compression Problem and Generalizations}
Wyner and Ziv \cite{wyner-ziv1976} derived an operational information definition for the lossy compression problem of Fig.~\ref{fg:blockdiagram} with respect to a single-letter fidelity of reconstruction, when the  joint  sequence of random variables (RVs) $\{(X_{t}, Y_{t}): t=1,2, \dots\}$ takes  values in  sets  of finite cardinality, $\{\cal{X},\cal{Y}\}$, and it is   generated independently according to the joint probability distribution function ${\bf P}_{X, Y}$.  Wyner \cite{wyner1978} generalized \cite{wyner-ziv1976}  to RVs $\{(X_{t}, Y_{t}): t=1,2, \dots\}$ that takes values in  abstract alphabet spaces $\{\cal{X},\cal{Y}\}$, and hence   include continuous-valued RVs. 
\begin{figure}
\centering
  \includegraphics[width=0.75\textwidth]{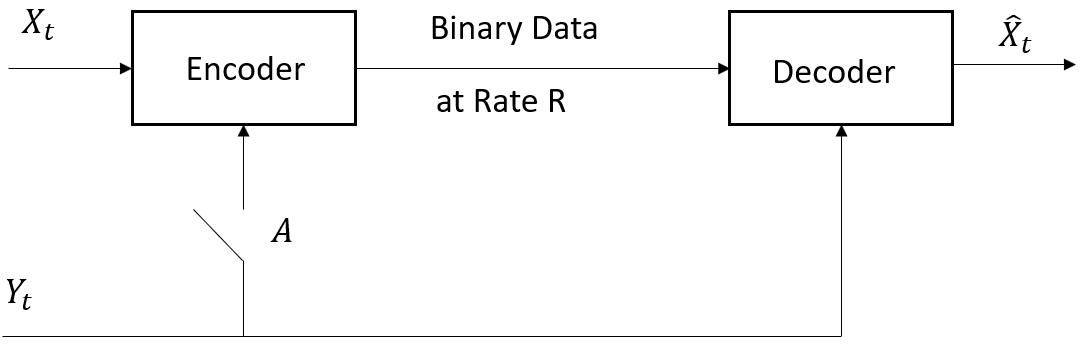}
 \caption{The Wyner and Ziv \cite{wyner-ziv1976} block diagram of lossy compression. If switch A is closed then  the side information is   available at both  the encoder and the decoder; if switch A is open the side information is  only available at the decoder.}
  \label{fg:blockdiagram} 
\end{figure}

{\it (A) Switch A Closed.} When   the side information $\{Y_{t}: t=1,2, \dots\}$ is available, noncausally,  at both  the encoder and the decoder,   Wyner \cite{wyner1978} (see also Berger \cite{berger:1971})  characterized  the infimum of all achievable operational rates (denoted by $\overline{R}_1(\Delta_X)$ in \cite{wyner1978}),  subject to a single-letter fidelity with average distortion less than or equal to $\Delta_X \in [0,\infty)$, by the single-letter operational information theoretic conditional rate distortion function (RDF):
\begin{align}
{R}_{X|Y}(\Delta_X) \tri& \inf_{{\cal M}_0(\Delta_X)}  I(X; \widehat{X}|Y), \hso \Delta_X \in [0,\infty)  \label{eq:OP1ck1}\\
=& \inf_{{\bf P}_{\widehat{X}|X,Y}:{\bf E}\big\{d_X(X, \widehat{X})\big\}\leq \Delta_X}  I(X; \widehat{X}|Y) \label{eq:OP1ck1_new}
\end{align}
where ${\cal M}_0(\Delta_X)$ is specified by the set 
\begin{align}
{\cal M}_0(\Delta_X)\tri & \Big\{ \widehat{X}: \Omega \rar \widehat{\cal X} : \hso  {\bf E}\Big\{d_X(X, \widehat{X})\Big\}\leq \Delta_X \Big\} \label{eq:OP2ck1}
\end{align}
and where $\widehat{X}$ is the reproduction of $X$, $I(X; \widehat{X}|Y)$ is the conditional mutual information between $X$ and $\widehat{X}$ conditioned on $Y$,  and $d_X(\cdot, \cdot)$ is the fidelity criterion between $x$ and $\widehat{x}$.  The infimum  in (\ref{eq:OP1ck1}) is over all joint distributions ${\bf P}_{X,Y, \widehat{X}}$ such that the marginal distribution ${\bf P}_{X,Y}$ is the joint distribution of the source $(X, Y)$. 

{\it (B) Switch A Open.} When   the side information $\{Y_{t}: t=1,2, \dots\}$ is available, noncausally,  only at the decoder, 
 Wyner \cite{wyner1978} characterized the infimum of all achievable operational rates (denoted by $R^*(\Delta_X)$ in \cite{wyner1978}),  subject to a single-letter fidelity with average distortion less than or equal to $\Delta_X$, by the  single-letter  operational information theoretic  RDF, as a function of  an auxiliary RV $Z: \Omega \rar {\cal Z}$:
\begin{align}
\overline{R}(\Delta_X) \tri& \inf_{ {\cal M}(\Delta_X)} \Big\{I(X;Z)-I(Y;Z)\Big\}, \; \Delta_X \in [0,\infty) \label{rdf_d1_a}\\
 =&\inf_{{\cal M}(\Delta_X)} I(X; Z|Y) \label{rdf_d1}
\end{align}
where  ${\cal M}(\Delta_X)$ is specified by the set of auxiliary RVs $Z$, 
\begin{align}
{\cal M}(\Delta_X)\tri& \Big\{ Z: \Omega \rar {\cal Z} : \hso {\bf P}_{Z|X,Y}={\bf P}_{Z|X},    \hso \exists \: \mbox{measurable function $f: {\cal Y}\times {\cal Z} \rar \widehat{\cal X}$}, \; \widehat{X}=f(Y,Z), \nonumber \\
&{\bf E}\big\{d_X(X, \widehat{X})\big\}\leq \Delta_X \Big\}. \label{rdf_d2}
\end{align}
Wyner's   realization of the joint  measure  ${\bf P}_{X, Y, Z, \widehat{X}}$ induced by the RVs $(X, Y, Z, \widehat{X})$,  is  illustrated in 
 Fig.~\ref{fg:onlydecoder}, where  $Z$ is the output of the ``test channel'', ${\bf P}_{Z|X}$. 

{\it Special Case of Switch A Open with Causal Side Information.} When   the side information $\{Y_{t}: t=1,2, \dots\}$ is causally available,  only at the decoder, it follows from \cite{wyner1978}, that the infimum of all achievable operational rates,  denoted by $R^{*,CSI}(\Delta_X)$,  is characterized by a degenerate version of $\overline{R}(\Delta_X)$, given by   
\begin{align}
\overline{R}^{CSI}(\Delta_X) \tri& \inf_{ {\cal M}(\Delta_X)}I(X;Z), \; \Delta_X \in [0,\infty). \label{rdf_d1_a_csi}
\end{align}

  Throughout \cite{wyner1978} the following  assumption   is imposed.

\begin{assumption}
\label{ass_1}
$
I(X;Y)<\infty $ (see \cite{wyner1978}).
\end{assumption}

For scalar-valued RVs $(X, Y, \widehat{X},Z)$ with square-error distortion, Wyner \cite{wyner1978} constructed the optimal realizations $\widehat{X}$ and $\widehat{X}=f(X, Z)$ that achieve the characterizations of the RDFs $R_{X|Y}(\Delta_X)$ and $\overline{R}(\Delta_X)$, respectively, and showed $\overline{R}(\Delta_X)= R_{X|Y}(\Delta_X)$. 

The main objective of this paper is to generalize Wyner's \cite{wyner1978}  optimal realizations $\widehat{X}$ and $\widehat{X}=f(X, Z)$ that achieve the  RDFs $R_{X|Y}(\Delta_X)$ and $\overline{R}(\Delta_X)$,  to  multivariate-valued RVs $(X, Y, \widehat{X},Z)$, and to show   $\overline{R}(\Delta_X)= R_{X|Y}(\Delta_X)$. Our main contribution  lies in the derivation of  structural properties of the optimal test channels that achieve the RDFs, and their realizations. Further, these  structural properties are indispensable in other problems of rate distortion theory. In particular, it is verified (see Remark~\ref{comment}) that the optimal realization of the test channel that achieves the RDF of the remote source coding problem\footnote{The RDF of the remote sensor problem is a generalization Wyner's RDF  $\overline{R}(\Delta_X)$, with the encoder observing a noisy version of the RVs generated by the source.}, given  in  \cite[Theorem~4 and Abstract]{tian-chen2009} and \cite[Theorem~3A]{Zahedi-Ostegraard-2014}, when specilized to scalar RVs, and  Wyner's RDF     $\overline{R}(\Delta_X)$, do not generate Wyner's value of the  RDF $\overline{R}(\Delta_X)$ and optimal test channel realization that achieves it, contrary to the current belief,  i.e.,  \cite{tian-chen2009} and \cite[Theorem~3A]{Zahedi-Ostegraard-2014}.
  The remote sensor problem\footnote{Remark~\ref{comment} implies that, for multivariate-valued Gaussian RVs,  the characterization of the RDF of the remote sensor problem and the test channel that achieves it, are currently  not known.}  is introduced by Draper and Wornell in \cite{draper-wornell2004}. 

{\it (C) Marginal RDF.}  If there is no  side information $\{Y_t: t=1, \ldots\}$, or the side information is  independent of the source $\{X_t: t=1, \ldots\}$, the RDFs ${R}_{X|Y}(\Delta_X), \overline{R}(\Delta_X)$ degenerate to   the marginal  RDF  $R_{X}(\Delta_{X})$,   defined by 
 \begin{align}
{R}_{X}(\Delta_X) \tri  \inf_{{\bf P}_{\widehat{X}|X}:{\bf E}\big\{d_X(X, \widehat{X})\big\}\leq \Delta_X }  I(X; \widehat{X}), \hso \Delta_X \in [0,\infty).  \label{eq:cl}
\end{align}

{\it (D) Gray's Lower Bounds.} Related to the RDF $R_{X|Y}(\Delta_X)$ is Gray's  \cite[Theorem~3.1]{gray1973} lower bound,  
\bea
R_{X|Y}(\Delta_X)\geq R_{X}(\Delta_X)-I(X;Y). \label{gray_lb_1}
\eea

{\it (E) The Draper and Wornell \cite{draper-wornell2004} Distributed Remote Source Coding Problem.}  Draper and Wornell \cite{draper-wornell2004} generalized the RDF $\overline{R}(\Delta_X)$, when the source  to be estimated at the decoder is  $S: \Omega \rar {\cal S}$, and it  is not directly   observed at the encoder. Rather, the encoder observes a RV $X:\Omega \rar {\cal X} $ (which is correlated with $S$), while  the decoder observes another RV, as side information,  $Y:\Omega \rar {\cal Y}$, which provides information on $(S,X)$.  The aim is to reconstruct $S$ at the decoder by $\widehat{S}: \Omega \rar \widehat{\cal S} $, subject to an average distortion ${\bf E}\{d_S(S,\widehat{S})\}\leq \Delta_S$, by a function   $\widehat{S}=f(Y,Z)$. \\
The RDF for this  problem, called distributed remote source coding problem,    is defined by \cite{draper-wornell2004}
 \begin{align}
\overline{R}^{PO}(\Delta_S)
 =&\inf_{{\cal M}^{PO}(\Delta_S)} I(X; Z|Y) \label{rdf_po_1}
\end{align}
where  ${\cal M}^{PO}(\Delta_S)$ is specified by the set of auxiliary RVs $Z$, 
\begin{align}
{\cal M}^{PO}(\Delta_S)\tri& \Big\{ Z: \Omega \rar {\cal Z} : \hso {\bf P}_{Z|S, X,Y}={\bf P}_{Z|X},    \hso \exists \: \mbox{measurable function $f^{PO}: {\cal Y}\times {\cal Z} \rar \widehat{\cal S}$}, \nonumber \\
&\widehat{S}=f^{PO}(Y,Z), \hso {\bf E}\big\{d_S(S, \widehat{S})\big\}\leq \Delta_S \Big\}. \label{rdf_po_2}
\end{align}
It should be mentioned that  if  $S=X-$a.s (almost surely) then  $\overline{R}^{PO}(\Delta_S)$ degenerates\footnote{This implies the optimal test channel that achieves  the characterization of the RDF $\overline{R}^{PO}(\Delta_S)$ should degenerate to the optimal test channel that achieves the characterization of the RDF $\overline{R}(\Delta_X)$.} to $\overline{R}(\Delta_X)$. \\
For scalar-valued jointly Gaussian RVs $(S, X, Y, Z, \widehat{X})$ with square-error distortion,  Draper and Wornell \cite[eqn(3) and Appendix~A]{draper-wornell2004} derived the characterization of the RDF $\overline{R}^{PO}(\Delta_S)$, and constructed   the optimal realization $\widehat{S}=f^{PO}(Y,Z)$ that achieves this characterization. 

In  \cite{tian-chen2009}  the authors 
 investigated the RDF $\overline{R}^{PO}(\Delta_S)$ for 
the    multivariate jointly Gaussian RVs $(S, X, Y, Z, \widehat{X})$, with square-error distortion, and derived 
a characterization for the RDF $\overline{R}^{PO}(\Delta_S)$ in  \cite[Theorem~4]{tian-chen2009} and \cite[Theorem~3A]{Zahedi-Ostegraard-2014} (see \cite[eqn(26)]{Zahedi-Ostegraard-2014}). However, as it apparent in Remark~\ref{comment}, when $S=X-$almost surely, and hence  $\overline{R}^{PO}(\Delta_S)=\overline{R}(\Delta_X)$, the optimal test channel realizations  that are  used to derive  \cite[Theorem~4]{tian-chen2009} and  \cite[Theorem~3A]{Zahedi-Ostegraard-2014}, when substituted into    the RDF $\overline{R}(\Delta_X)$, i.e.,   (\ref{rdf_d1_a}), do not produce  Wyner's  \cite{wyner1978} value of the the RDF,  test channel realization  (and also do not  produce the known RDF and test channel realization of memoryless sources).
 This observation is sufficient to raise 
 concerns regarding  the validity of the water-filling solution given in  \cite[Theorem~4]{tian-chen2009} and \cite[Theorem~3A]{Zahedi-Ostegraard-2014}.

\begin{figure}
\centering 
 \includegraphics[width=0.65\textwidth]{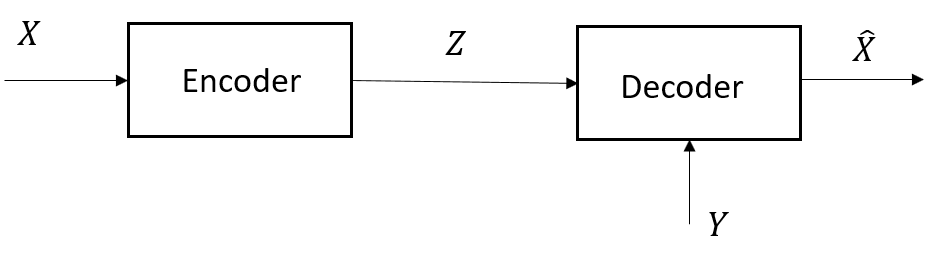}
 \caption{Test channel when side information is only available to the  decoder}
\label{fg:onlydecoder}
\end{figure}

\subsection{Problem Statement and Main Contributions}
\par In this paper we consider      a tuple of jointly independent and identically distributed multivariate Gaussian random
variables (RVs) $(X^n, Y^n)= \{(X_{t}, Y_{t}): t=1,2, \ldots,n\}$, with respect to the square-error fidelity, as defined below. 
\begin{align}
&X_{t} : \Omega \rightarrow {\mathbb R}^{n_x}= {\cal X}, \;  Y_{t} : \Omega \rightarrow {\mathbb R}^{n_y}={\cal Y},\; t=1,2, \ldots, n,\label{prob_1} \\
& X_t \in N(0, Q_X), \hst  Y_t \in N(0, Q_Y),  \label{prob_2}\\
&Q_{(X_t, Y_t)} =  {\mathbb E} \Big\{ \left[ \begin{array}{c} X_t \\ Y_t  \end{array} \right]  \left[ \begin{array}{c} X_t \\ Y_t  \end{array} \right]^T \Big\}=  \left[ \begin{array}{cc} Q_X & Q_{X,Y} \\ Q_{X,Y}^T & Q_Y \end{array} \right], \\
&{\bf P}_{X_t, Y_t}={\bf P}_{X,Y} \hso \mbox{multivariate Gaussian distribution}, 
\label{prob_7}\\
& \widehat{X}_t: \Omega  \rar {\mathbb R}^{n_x}= {\cal  X}, \hso 
 t = 1,2 \ldots, n,\label{prob_8} \\
& D_{X} (x^n, \widehat{x}^n)= \frac{1}{n} \sum_{t=1}^n ||x_{t}-\widehat{x}_{t}||_{{\mathbb R}^{n_x}}^2\label{prob_4}
\end{align}
where $n_x, n_y$ are arbitrary positive integers, $I_{n_y}$ is the $n_y \times n_y$ diagonal matrix,  $X \in N(0,Q_X)$ means $X$ is a Gaussian RV,  with zero mean and covariance matrix $Q_X$,   and $||\cdot||_{{\mathbb R}^{n_x}}^2$ is the  Euclidean distance on ${\mathbb R}^{n_x}$.\\
To give additional insight we often consider the following  realization of side information\footnote{The condition $DD^T \succ 0$ ensures $I(X;Y)<\infty$, and hence Assumption~\ref{ass_1} is respected.}. 
\begin{align}
&Y_t = CX_t + DV_t,      \label{eq:sideInfo}\\
&V_t \in N(0, Q_V),\label{prob_3} \\
&C\in \mathbb{R}^{n_y\times n_x}, \hso  D \in \mathbb{R}^{n_y\times n_y}, \hso D D^T \succ 0,  \hso Q_V=I_{n_y} \label{prob_9_a}\\
&V^n \hso \mbox{independent of $X^n$}.\label{prob_9}
\end{align}
For the above specification of the source and distortion criterion,  we  derive the following results.
 
\begin{enumerate}
\item  {\bf  Theorem~\ref{thm_rw}, Fig.~\ref{fg:realization}.} Structural properties   of  optimal realization  of  $\widehat{X}$  that achieves the RDFs, $R_{X|Y}(\Delta_X)$,  closed form expression for $R_{X|Y}(\Delta_X)$.

\item {\bf  Theorem~\ref{thm:dec}.} Structural properties   of  optimal realization  of  $\widehat{X}$ and $\widehat{X}=f(Y,Z)$ that achieve the RDF,  $\overline{R}(\Delta_X)$,   and  closed form expressions for  $\overline{R}(\Delta_X)$.

\item  A proof that $\overline{R}(\Delta_X)$ and $R_{X|Y}(\Delta_X)$ coincide, calculation of positive surface such that Gray's lower bound (\ref{gray_lb_1}) holds with equality, and a proof  that the optimal test channel realization for the remote sensor problem,   that is used to derive  \cite[Theorem~4]{tian-chen2009} is incorrect (Remark~\ref{comment}).
\end{enumerate}

In  Remark~\ref{rem-wyner},  we consider the tuple of scalar-valued, jointly Gaussian RVs $(X, Y)$,  with square error distortion function, and verify that our optimal realizations  of  $\widehat{X}$ and closed form expressions for $R_{X|Y}(\Delta_X)$ and $\overline{R}(\Delta_X)$ are identical to  Wyner's \cite{wyner1978} realizations  and RDFs.

We emphasize that  past literature often deals with the calculation of RDFs using optimization techniques, without much emphasis on the  structural properties of the realizations of the test channels, that achieve the characterizations of the RDFs. Because of this, often the optimization problems appear intractable, while closed form solutions are rare. It will be indefensible to claim that solving an optimization problem of a  RDF, without specifying the realization of the optimal test channel that achieves the  value of the RDF,   fully characterizes the RDF. 
As demonstrated   by Wyner \cite{wyner1978} for a tuple of scalar jointly Gaussian RVs $(X,Y)$ with square-error distortion criterion,  the identity $\overline{R}(\Delta_X)={R}_{X|Y}(\Delta_X)$ holds, because the realizations that achieve these RDFs are explicitly constructed. 
Although, in the current paper the emphasis is on 1), 2) above, our derivations are generic and bring new insight into the construction of optimal test channels for  other distributed source coding problems.




\subsection{Additional Literature Review}
The  formulation of Fig.~\ref{fg:blockdiagram} is generalized to  other multiterminal or distributed  lossy compression problems, 
 such as, relay networks, sensor networks   etc.,
 under various code formulations and  assumptions.
 Oohama \cite{Oohama1997} analyzed the lossy compression problems for a tuple of scalar  correlated  Gaussian memoryless  sources with square error distortion criterion, and
determined the rate-distortion region,  in the special case when one
source provides partial side information to the other source. 
 Oohama \cite{Oohama2005} analyzed the separate lossy compression 
problem for $L+1$ scalar correlated Gaussian memoryless sources, when $L$ act as  sources partial side information
at the decoder for the reconstruction of the remaining source, and gave a partial answer to the  rate distortion region. Oohama \cite{Oohama2005} also proved that his problem gives as a special case the additive white Gaussian CEO problem analyzed by Viswanathan and Berger \cite{ViswanathanCEO1997}. In addition, Ekrem and Ulukus \cite{SUlukus2012} and Wang and Chen \cite{JunChen2014} expanded Oohama\cite{Oohama2005} main results, by  deriving  an outer bound on  the rate region of the vector Gaussian multiterminal source. 
%
The vast literature  on multiterminal or distributed lossy compression of jointly Gaussian sources with square-error distortion (mentioned above), is often confined to a tuple of correlated  RVs $X: \Omega \rar {\mathbb R}, Y: \Omega \rar {\mathbb R}$. The above literature treats the optimization problems of RDFs without much emphasis on the structural properties of the  optimal test channels that achieve the characterizations of the RDFs. 

\subsection{Main Theorems of the Paper}
The characterizations of the RDFs $R_{X|Y}(\Delta_X)$ and $\overline{R}(\Delta_X)$ are  encapsulated in  Theorem~\ref{thm_rw} and Theorem~\ref{thm:dec}, stated below. These theorems include,  structural properties of optimal test channels or realizations of $\widehat{X}$, that induce joint  distributions, which  achieve the RDFs, and closed form expressions of the RDFs based on a water-filling. The realization of the optimal test channel of $R_{X|Y}(\Delta_X)$ is shown in Fig.~\ref{fg:realization}.

First, we introduce some notation. We denote  the  covariance of $X$ and $Y$ by   
\begin{align}
 Q_{X,Y} \tri\Cov\Big(X,Y\Big). 
\end{align}
We denote  the  covariance of $X$ conditioned on  $Y$ by,  
\begin{align}
Q_{X|Y} & \tri  \nonumber \Cov(X,X |Y)\\
& =  {\bf  E} \Big\{  \Big(X - {\bf E}\big(X\Big|Y\big) \Big)  \Big(X - {\bf E}\big(X\Big|Y\big) \Big)\T \Big\} \hso \mbox{if $(X, Y)$ is jointly Gaussian.}
\end{align}
where the second equality is due to a property of jointly Gaussian RVs. 

The first  theorem gives the optimal test channel that achieves the characterization of the RDF $R_{X|Y}(\Delta_X)$, and its water-filling representation. 

\begin{theorem}
Characterization and water-filling solution of $R_{X|Y}(\Delta_X)$\\ 
\label{thm_rw}
Consider the RDF $R_{X|Y}(\Delta_X)$ defined by  (\ref{eq:OP1ck1}), for the multivariate Gaussian source  with mean-square error distortion defined by (\ref{prob_1})-(\ref{prob_9}).\\
Then the following hold.

(a) The optimal  realization $\widehat{X}$  that achieves $R_{X|Y}(\Delta_X)$ is represented
 by 
\begin{align}
\widehat{X} &= H X + \Big(I_{n_x}- H\Big)Q_{X,Y}Q_Y^{-1} Y +  W \label{eq:realization_sp} \\
&= H X + \Big(I_{n_x}- H\Big)Q_{X,Y}Q_Y^{-1} Y + H \Psi \label{eq:realization_sp_new} 
\end{align}
where
   \begin{align}
 &H  \tri I_{n_x} - \Sigma_{\Delta} Q_{X|Y}^{-1}=I_{n_x} - Q_{X|Y}^{-1} \Sigma_{\Delta}=H\T \succeq 0,  \label{eq:realization_nn_1_sp}  \\
 &W = H \Psi, \hso \Psi \in N(0, Q_\Psi), \hso Q_\Psi \tri \Sigma_\Delta H^{-1}=H^{-1} \Sigma_\Delta, \\
   & Q_W \tri H \Sigma_\Delta=  \Sigma_\Delta -\Sigma_\Delta Q_{X|Y}^{-1} \Sigma_\Delta= \Sigma_{\Delta}  H \succeq 0,\\
  & \Sigma_{\Delta} \tri  {\bf E} \Big\{ \Big( X - \widehat{X} \Big) \Big(X - \widehat{X} \Big)\T \Big\}, \\
 &Q_{X|Y}= Q_X -Q_{X,Y} Q_Y^{-1} Q_{X,Y}\T, \ \ Q_{X,Y}= Q_X C\T, \ \ Q_Y= C Q_X C\T + D  D\T    \label{eq:realization_nn_1_sp_new}
\end{align}
Moreover, the following structural properties hold:\\
(1)  The optimal test channel satisfies  
\begin{align}
&(i) \hso {\bf P}_{X|\widehat{X}, Y}={\bf P}_{X|\widehat{X}}, \label{pr_1} \\
&(ii)    \hso {\bf E}\Big\{X\Big|\widehat{X}, Y\Big\}={\bf E}\Big\{X\Big|\widehat{X}\Big\}=\widehat{X} \hso \Longrightarrow \hso {\bf E}\Big\{X\Big|Y\Big\}= {\bf E}\Big\{\widehat{X}\Big|Y\Big\}. \label{pr_1_a}
\end{align} 
(2) The matrices 
\begin{align}
 \big\{\Sigma_\Delta,& Q_{X|Y}, H, Q_W\big\} \hso \mbox{have  spectral} \nonumber \\
&   \mbox{decompositions w.r.t the same unitary matrix $U U\T=I_{n_x}, U\T U=I_{n_x}$.} \label{spe_d}
\end{align}
\par (b)  The RDF $ R_{X|Y}(\Delta_X)$ is given by the water-filling solution:
\begin{equation}
R_{X|Y}(\Delta_X) =\frac{1}{2}  \log \max\big\{1, \det(Q_{X|Y}\Sigma_\Delta^{-1})\big\}=\frac{1}{2} \sum_{i=1}^{n_x} \log \frac{\lambda_{i}}{\delta_{i}}  \label{thm_wf_1}
\end{equation}
where 
\begin{equation}
{\bf E}\big\{||X-\widehat{X}||_{{\mathbb R}^{n_x}}\big\}=\trace\big(\Sigma_\Delta\big)= \sum_{i=1}^{n_x} \delta_{i} = \Delta_X, \hst \delta_{i}= \left\{ \begin{array}{lll} \mu, & \mbox{if} & \mu < \lambda_i \\ \lambda_i, & \mbox{if} & \mu \geq \lambda_i \end{array} \right. \label{thm_wf_2}
\end{equation}
and where $\mu \in [0,\infty)$ is a Lagrange multiplier (obtained from the Kuch-Tucker conditions), and
\begin{align}
 Q_{X|Y}&= U\Lambda U\T, \ \ \Lambda =\diag{\{\lambda_{1},\dots, \lambda_{n_x}\}}, \ \ \lambda_{1} \geq \lambda_2 \geq \dots \geq \lambda_{n_x}  \\ 
\Sigma_{\Delta} &= U \Delta U\T,\ \ \Delta = \diag{\{\delta_{1},\ldots, \delta_{{n_x}}\}}, \ \  \delta_{1}\geq \delta_{2} \geq \ldots \geq  \delta_{{n_x}}   .
\end{align}
\par (c) The optimal  $\widehat{X}$ of part (a)   that achieves $R_{X|Y}(\Delta_X)$ is realized by the parallel channel scheme  depicted  in Fig.~\ref{fg:realization}.
\par (d) If $X$ and $Y$ are independent or   $Y$ is replaced by  a RV that generates the trivial information, i.e., the $\sigma-$algebra of $Y$, is $\sigma\{Y\}=\{\Omega, \emptyset\}$ (or $C=0$ in  (\ref{eq:sideInfo})), then (a)-(c) hold with   $Q_{X|Y}=Q_X, Q_{X,Y}=0$,  and $R_{X|Y}(\Delta_X)=R_X(\Delta_X)$, i.e. becomes the marginal RDF of $X$.
\end{theorem}

\begin{figure}
\centering  \includegraphics[scale=0.5]{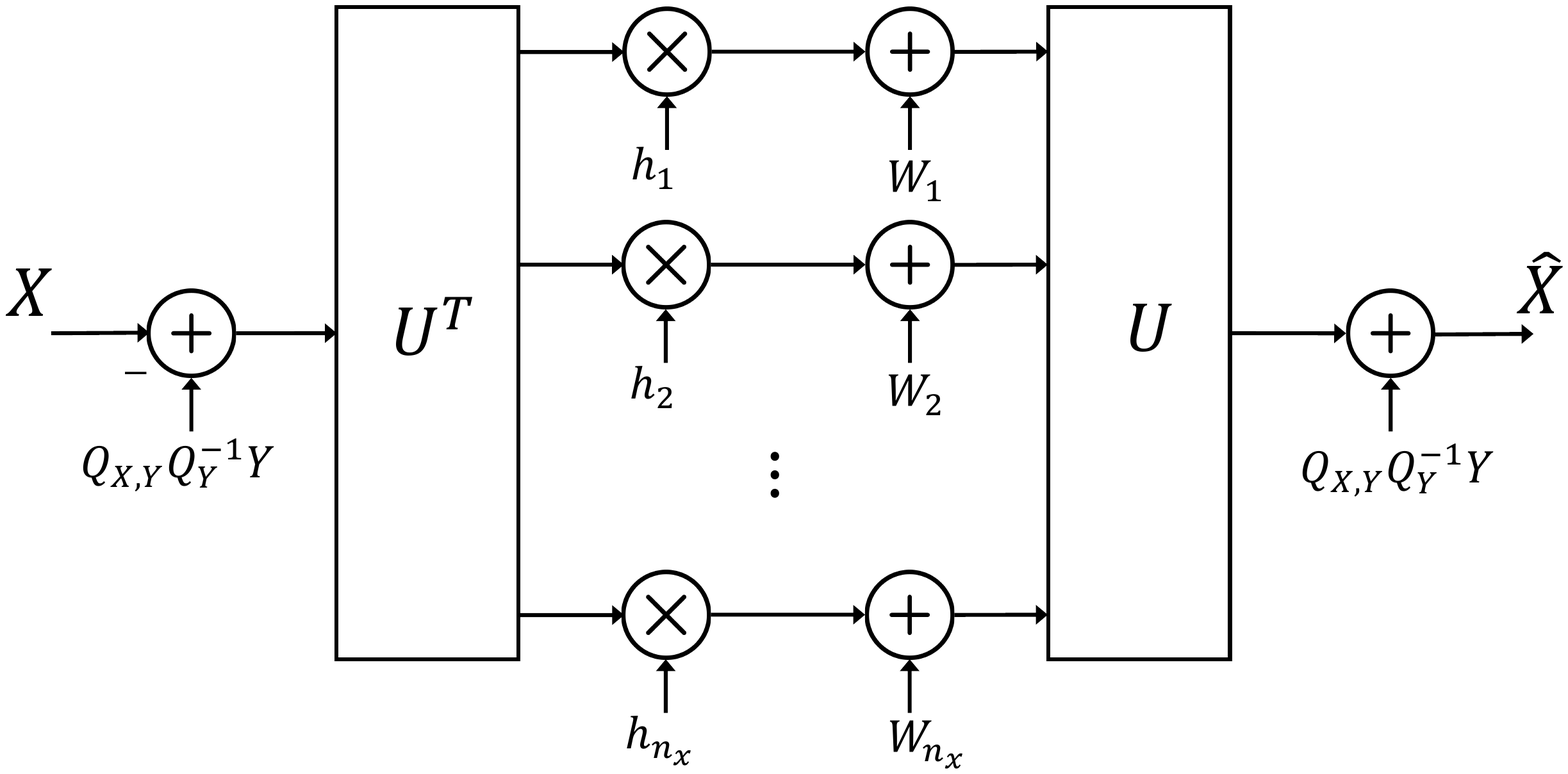}
 \caption{$R_{X|Y}(\Delta_X)$: A realization of optimal reproduction $\widehat{X}$ over parallel additive Gaussian noise channels, where $h_{i}\tri 1-\frac{\delta_{i}}{\lambda_i}\geq 0, i=1, \ldots, n_x$ are the diagonal element of  the spectral decomposition of the  matrix $H=U \diag\{h_1, \ldots, h_{nx}\} U\T$, and $W_i\in N(0,h_i\delta_{i}),  i=1, \ldots, n_x$  .}
 \label{fg:realization}
\end{figure}

The proof of Theorem \ref{thm_rw} is given Section 
\ref{mainres}, and it is based on the derivation of the  structural properties. Some of the implications are briefly described below. \\

{\it Conclusion 1.} The construction and the structural properties of the optimal test channel ${\bf P}_{X|\widehat{X}, Y}$ that achieves the water-filling characterization of the RDF $R_{X|Y}(\Delta_X)$  of  Theorem \ref{thm_rw}, are  not documented elsewhere in the literature.

(i) Structural property (\ref{pr_1}) strengthens Gray's \cite[Theorem~3.1]{gray1973} (see proof of (\ref{gray_lb_1})), inequality, 
\bea
I(X; \widehat{X}|Y)\geq  I(X; \widehat{X})-I(X;Y).
\eea
to the equality  
\begin{align}
I(X; \widehat{X}|Y)=& I(X; \widehat{X})-I(X;Y)\in [0,\infty) \hst \mbox{if} \hso  {\bf P}_{X|\widehat{X}, Y}={\bf P}_{X|\widehat{X}} \\
  =&\frac{1}{2}  \log\big\{ \det(Q_{X|Y}\Sigma_\Delta^{-1})\big\},\hso Q_{X|Y}-\Sigma_{\Delta}\succeq 0, \hso {\bf E}\big\{||X-\widehat{X}||_{{\bf R}^{n_x}}^2\big\}=\trace\big(\Sigma_\Delta\big)\leq \Delta_X
\end{align}  
Structural property (\ref{pr_1_a}) means the  subtraction of equal quantities  ${\bf E}\Big\{X\Big|Y\Big\}$ and ${\bf E}\Big\{\widehat{X}\Big|Y\Big\}$ at the encoder and decoder, respectively, without affecting the information measure, see Fig.~\ref{fg:realization}. \\
Theorem~\ref{thm_rw}.(b), (c), are   obtained, with the aid of part (a), and Hadamard's inequality that shows $Q_{X|Y}$ and $\Sigma_\Delta$ have the same eigenvectors.\\
  Structural propery   (\ref{pr_1}) implies that    Gray's  \cite[Theorem~3.1]{gray1973}  lower bound (\ref{gray_lb_1}) holds with equality for a strictly positive surface\footnote{See Gray \cite{gray1973} for definition.}  $\Delta_X \leq {\cal D}_C(X|Y)\subseteq [0,\infty)$, i.e., 
\bea
R_{X|Y}(\Delta_X)= R_{X}(\Delta_X)-I(X;Y),\hst  \Delta_X \leq  
{\cal D}_C(X|Y)\tri \big\{\Delta_X\in [0,\infty): \Delta_X \leq n_x \lambda_{n_x}\big\}.
\eea
 That is,  the set ${\cal D}_C(X|Y)$ excludes values of   $\Delta_X\in [0,\infty)$ for which   water-filling is active in   (\ref{thm_wf_1}), (\ref{thm_wf_2}).

(ii) Structural property (2), i.e.,  the matrices $\{\Sigma_\Delta, Q_{X|Y}, H, Q_W\}$ are nonnegative symmetric, and have a spectral decomposition with respect to  the same unitary matrix $U U\T=I_{n_x}$\cite{Horn:2013}, implies that  the test channel is equivalently represented by parallel additive Gaussian noise channels (subject to pre-processing and post-processing at the encoder and decoder). 

(iii) Remark~\ref{rem-wyner} shows that 
 the realization of  optimal $\widehat{X}$ of   Fig.~\ref{fg:realization} that achieves the RDF of  Theorem~\ref{thm_rw}, degenerates to    Wyner's \cite{wyner1978} optimal realization that achieves the  RDF $R_{X|Y}(\Delta_X)$, for the tuple of scalar-valued, jointly Gaussian RVs $(X, Y)$,  with square error distortion function. 

The second  theorem gives the optimal test channel that achieves the characterization  of the RDF $\overline{R}(\Delta_X)$, and further states that, there is no loss of compression rate if side information is only available at the decoder. That is, although in general,   $\overline{R}(\Delta_X)\geq R_{X|Y}(\Delta_X)$, an optimal  reproduction $\widehat{X}=f(Y,Z)$  of $X$,   where $f(\cdot, \cdot)$ is linear, is contructed such that the  inequality holds with equality.

\begin{theorem} Characterization and water-filling solution of  $\overline{R}(\Delta_X)$\\
\label{thm:dec}
Consider the RDF $\overline{R}(\Delta_X)$ defined by  (\ref{rdf_d1}), for the multivariate Gaussian source  with mean-square error distortion, defined by (\ref{prob_1})-(\ref{prob_9}).\\
Then the following hold. 

(a)  The characterization of the RDF, $\overline{R}(\Delta_X)$,  satisfies  
\begin{align}
\overline{R}(\Delta_X) \geq R_{X|Y}(\Delta_X) \label{lower_b_d_1}
\end{align}
where $R_{X|Y}(\Delta_X)$ is  given in Theorem~\ref{thm_rw}.(b). 

(b) The optimal  realization $\widehat{X}=f(Y, Z)$  that achieves the lower bound in (\ref{lower_b_d_1}), i.e.,  $\overline{R}(\Delta_X)=R_{X|Y}(\Delta_X)$,  is represented by 
\begin{align}
\widehat{X}=&f(Y,Z) \label{real_d1}\\
=&\Big(I-H\Big)Q_{X,Y}Q_Y^{-1} Y +Z, \\
Z=& H \Big(X + H^{-1} W\Big), \\
(H,& Q_W)  \hso \mbox{given by (\ref{eq:realization_nn_1_sp})-(\ref{eq:realization_nn_1_sp_new}), and (\ref{spe_d}) holds} \label{real_d2}.
\end{align}
Moreover, the following structural properties hold:\\
(1)  The optimal test channel satisfies 
\begin{align}
&(i) \hso  {\bf P}_{X|\widehat{X},Y,Z}={\bf P}_{X|\widehat{X},Y}={\bf P}_{X|\widehat{X}}, \label{str_p_1} \\
&(ii)   \hso {\bf E}\Big\{X\Big|\widehat{X}, Y,Z\Big\}={\bf E}\Big\{X\Big|\widehat{X}\Big\}=\widehat{X} \hso \Longrightarrow \hso {\bf E}\Big\{X\Big|Y\Big\}= {\bf E}\Big\{\widehat{X}\Big|Y\Big\}.\label{str_p_2}
\end{align} 
(2) Structural property (2) of Theorem~\ref{thm_rw}.(a) holds.  
\end{theorem}

The proof of Theorem \ref{thm:dec} is given Section 
\ref{mainres}, and it is based on  the derivation of  the  structural properties and  Theorem~\ref{thm_rw}. Some implications are  discussed below.\\

{\it Conclusion 2.} The optimal reproduction $\widehat{X}=f(X,Z)$ or test channel distribution ${\bf P}_{X|\widehat{X},Y,Z}$ that achieve  the RDF $\overline{R}(\Delta_X)$  of Theorem \ref{thm:dec}, are  not reported in the literature. 

(i) From the structural property (1) of Theorem~\ref{thm:dec}, i.e., (\ref{str_p_1}) then follows the lower bound $\overline{R}(\Delta_X) \geq R_{X|Y}(\Delta_X)$ is achieved by the realization $\widehat{X}=f(Y,Z)$ of  Theorem \ref{thm:dec}.(b), i.e., for a given $Y=y$, then $\widehat{X}$ uniquely defines $Z$.   

(ii) If $X$ is  independent of $Y$ or $Y$ generates a trivial information,  then the RDFs $\overline{R}(\Delta_X)=\overline{R}_{X|Y}(\Delta_X)$  degenerate to the classical RDF of the source $X$, i.e., $R_{X}(\Delta_X)$, as expected. This is easily verified   from  (\ref{real_d1}), (\ref{real_d2}), i.e., $Q_{X,Y}=0$ which implies $\widehat{X}=Z$. \\
For scalar-valued RVs,  $X : \Omega \rightarrow {\mathbb R}, Y : \Omega \rightarrow {\mathbb R}, X \in N(0, \sigma_X^2)$, and $X$ independent of $Y$, then  the optimal realization reduces to 
\begin{align}
&\widehat{X}=Z= \Big(1-\frac{\Delta_X}{\sigma_X^2}\Big) X + \sqrt{\Big(1-\frac{\Delta_X}{\sigma_X^2}\Big)\Delta_X} \overline{W}, \hso \overline{W}\in N(0,1), \hso \sigma_X^2 \leq \Delta_X, \label{mem_scal_1}   \\
&Q_{\widehat{X}}=Q_{Z}=\sigma_{\widehat{X}}^2=\sigma_X^2 -\Delta_X.\label{mem_scal_2}
\end{align}
 as expected.


(iii) Remark~\ref{rem-wyner} shows that 
 the realization  of optimal $\widehat{X}=f(Y,Z)$  that achieves the RDF $\overline{R}(\Delta_X)$ of   Theorem \ref{thm:dec}, degenerates  to Wyner's \cite{wyner1978} realization that achieves the  RDF $\overline{R}(\Delta_X)$,  of the tuple of scalar-valued, jointly Gaussian RVs $(X, Y)$,  with square error distortion function. 

(iv) Remark~\ref{comment} shows that,  when specialized to Wyner's  RDF  $\overline{R}(\Delta_X)$,   the optimal test channel realizations that achieve the  RDFs of the distributed remote source coding problems  in   \cite[Theorem~4]{tian-chen2009},
   does not degenerate  to  Wyner's optimal test channel realization, $(\widehat{X}, Z)$, that achieves  the RDF  $\overline{R}(\Delta_X)$, contrary to what is expected \cite[Abstract]{tian-chen2009}.

The next corollary follows from the above two theorems.

\begin{corollary} Characterization  of  $\overline{R}^{CSI}(\Delta_X)$\\
\label{cor:c-dec}
Consider the RDF  $\overline{R}^{CSI}(\Delta_X)$ defined by  (\ref{rdf_d1_a_csi}), for the multivariate Gaussian source  with mean-square error distortion, defined by (\ref{prob_1})-(\ref{prob_9}).\\
The optimal test channel of the RDF $\overline{R}^{CSI}(\Delta_X)$  is induced by the realization $Z= H \Big(X + H^{-1} W\Big)$, where 
$(H,Q_W)$   are given by (\ref{eq:realization_nn_1_sp})-(\ref{eq:realization_nn_1_sp_new}) of  Theorem \ref{thm:dec}, and
\begin{align}
\overline{R}^{CSI}(\Delta_X)=&\inf_{{\cal Q}(\Delta_X) }  \Big\{H(X)- H(X|Z)\Big\}\\ 
=& \inf_{{\cal Q}(\Delta_X) }   \frac{1}{2}\log\Big\{ \det (Q_XQ_{X|Z}^{-1})\Big\}
\end{align}
where   $(H, Q_W)$ are   given by (\ref{eq:realization_nn_1_sp})-(\ref{eq:realization_nn_1_sp_new}), and 
\begin{align}
&{\cal Q}(\Delta_X)\tri \bigg\{\Sigma_{\Delta} \succeq 0:  \trace \big( \Sigma_{\Delta} \big)\leq{\Delta_X}, \hst  HQ_X H\T \Big(HQ_XH\T +Q_W\Big)^{-1}HQ_X H\T \preceq  Q_X \bigg\}\\
&Q_{X|Z}=Q_X-HQ_X H\T\Big(HQ_X H\T +Q_W\Big)^{-1} HQ_X H\T.
\end{align}
\end{corollary}

The rest of the paper is organized as   follows. In Section~\ref{sect:wyner} we  review of Wyner's \cite{wyner1978} operational definition of lossy compression and state a fundamental theorem on mean-square estimation that we use throughout the paper.
  In Section~\ref{sect:proofs} we  prove the structural properties and  the  two  main theorems.
\section{Preliminaries}
\label{sect:wyner}
\par In this section we review  the Wyner   \cite{wyner1978} source coding problems with  fidelity of Fig, \ref{fg:blockdiagram}.

 We begin with the notation, which follows closely \cite{wyner1978}.

\subsection{Notation}
Let  ${\mathbb Z} \tri \{\ldots, -1,0,1,\ldots \}$ the set of all integers, ${\mathbb N} \tri \{0, 1,2, \ldots, \}$ the set of natural integers, ${\mathbb Z}_+ \tri \{1,2, \ldots, \}$. For $n \in {\mathbb Z}_+$ denote the following finite subset of the above defined set,  ${\mathbb Z}_n\tri \{1,2,\ldots, n\}$. 

Denote the real numbers by
${\mathbb R}$ and the set of
positive and of strictly positive real numbers, respectively, by
${\mathbb R}_+ = [0,\infty)$ and ${\mathbb R}_{++}=(0,\infty)$. For any  matrix $A\in \mathbb{R}^{p\times m}, (p,m)\in {\mathbb Z}_+\times {\mathbb Z}_+$, we denote its transpose by $A\T$, and for $m=p$,  we denote its trace  by  $\trace(A)$,  and by  $\diag\{A\}$, the matrix with diagonal entries $A_{ii},~i=1,\ldots,p$,  and zero elsewhere. The identity matrix with dimensions $p\times p$ is designated as $I_p$.

Denote an arbitrary set or space by 
   ${\cal U}$ and the product space formed by $n$ copies of it by  ${\cal U}^n \tri \times_{t=1}^n {\cal U}$. $u^n \in {\cal U}^n$ denotes the set of $n-$tuples $u^n \tri (u_1,u_2, \ldots, u_n)$, where $u_k \in {\cal U}, k=1, \ldots, n$ are its coordinates.
   
Denote a probability space by $(\Omega, {\cal F}, {\mathbb P})$. For a sub-sigma-field  ${\cal G} \subseteq {\cal F}$, and $A \in {\cal F}$, denote by ${\mathbb P}(A|{\cal G})$ the conditional probability of $A$ given ${\cal G}$, i.e., ${\mathbb P}(A|{\cal G})={\mathbb P}(A|{\cal G})(\omega), \omega \in \Omega$ is a measurable function on $\Omega$. 

 On the above probability space,  consider two-real valued random variables (RV) $X: \Omega \rar {\cal X}, Y: \Omega \rar {\cal X}$, where $({\cal X}, {\cal B}({\cal X})),    ({\cal Y}, {\cal B}({\cal Y}))$ are arbitrary measurable spaces.  The measure (or joint distribution if ${\cal X}, {\cal Y}$ are Euclidean spaces) induced by $(X, Y)$ on ${\cal X} \times {\cal Y}$ is denoted by ${\bf P}_{X,Y}$ or  ${\bf P}(dx,dy)$ and their marginals on ${\cal X}$ and ${\cal Y}$ by ${\bf P}_X$ and ${\bf P}_Y$, respectively. The conditional measure of RV $X$ conditioned on $Y$ is denoted by ${\bf P}_{X|Y}$ or ${\bf P}(dx|y)$, when $Y=y$ is fixed.  

On the above probability space, consider  three-real values RVs $X: \Omega \rar {\cal X}, Y: \Omega \rar {\cal X}$, $Z: \Omega \rar {\cal Z}$. We say that RVs $(Y, Z)$ are conditional independent given RV $X$ if ${\bf P}_{Y, Z|X}={\bf P}_{Y|X} {\bf P}_{Z|X}-$a.s (almost surely) or equivalently ${\bf P}_{Z|X, Y}={\bf P}_{Z|X}-$a.s; the specification a.s is often omitted.  We often  denote the above conditional independence by the Markov chain (MC) $Y \leftrightarrow X \leftrightarrow Z$. 

Finally, for RVs $X, Y$ etc. $H(X)$ denotes differential entropy of $X$, $H(X|Y)$ conditional differential entropy of $X$ given $Y$, $I(X;Y)$ the mutual information between $X$ and $Y$, as defined in standard books on information theory, \cite{Gallager:1968}, \cite{pinsker:1964}. 

The notation $X \in N(0, Q_X)$ means $X$ is a Gaussian distributed RV with zero mean and covariance $Q_X\succeq 0$, where $Q_X \succeq 0$ (resp. $Q_X \succ 0 $) means $Q_X$ is  positive semidefinite (resp. positive definite).


\subsection{Wyner's Coding Theorems with Side Information at the Decoder} \label{pr:2}
\par For the sake of completeness, we introduce certain results from Wyner's paper \cite{wyner1978}, that we use in this paper.   

On a probability space $(\Omega, {\cal F}, {\mathbb P})$, consider  a tuple of jointly independent and identically distributed RVs $(X^n, Y^n)= \{(Y_{t}, Y_{t}): t=1,2, \ldots,n\}$,  
\begin{align}
X_{t} : \Omega \rightarrow {\cal Y}, \hso  Y_{t} : \Omega \rightarrow {\cal  Y}, \ \ t = 1,2, \ldots, n \label{jgrv}
\end{align}
with induced distribution ${\bf P}_{X_t, Y_t}={\bf P}_{X,Y}, \forall t$. 
Consider also the measurable function $d_{X}: {\cal X} \times \widehat{\cal X} \rar [0,\infty)$, for a measurable space $\widehat{\cal X}$. 
 Let   
 \bea
 {\cal I}_M \tri \big\{0,1, \ldots, M-1\big\}, \hso M \in {\mathbb Z}_M. 
 \eea
be a finite set.


\par A code $(n, M,D_X)$, when switch $A$ is open in Fig.~\ref{fg:blockdiagram},   is defined by two measurable  functions,  the encoder $F_E$ and the decoder $F_D$, with average distortion,   as follows. 
\begin{align*}
&F_{E} : {\cal X}^n \longrightarrow {\cal I}_M, \hst 
F_D: {\cal I}_M \times {\cal Y}^n \longrightarrow {\cal \widehat{\cal X}}^n, \\
&\frac{1}{n}  \mathbb{\bf E} \Big\{\sum_{t=1}^n d_X(X_t,\widehat{X}_t) \Big\} = D_X 
\end{align*}
where $\widehat{X}^n$ is again a sequence of RVs,  $\widehat{X}^n= F_D(Y^n,F_E(X^n))\in {\cal \widehat{\cal X}}^n$. A non-negative rate distortion pair $(R, \Delta_X)$ is said to be {\it achievable} if for every  $\epsilon >0$, and $n$ sufficiently large there exists a code $(n, M,D_X)$ such that 
\begin{equation*}
M \leq 2^{n(R + \epsilon)}, \hspace{0.4cm} D_X \leq\Delta_X+ \epsilon
\end{equation*}
Let ${\cal R}$ denote the set of all achievable pairs  $(R,\Delta_X)$, and define,  for  $ \Delta_X\geq 0$,  the infimum  of all achievable rates by 
\begin{equation}
R^*(\Delta_X) = \inf_{(R,\Delta_X) \in{\cal R}}R
\end{equation}
If for some $\Delta_X$ there is no $R <\infty$ such that $(R,\Delta_X) \in{\cal R}$, then set  $R^*(\Delta_X) =\infty.$  For arbitrary abstract spaces  Wyner \cite{wyner1978}  characterized the infimum of all achievable rates $R^*(\Delta_X)$, by the single-letter RDF, $\overline{R}(\Delta_X)$ given by (\ref{rdf_d1}), (\ref{rdf_d2}),    in terms of an auxiliary RV $Z: \Omega \rar {\cal Z}$. Wyner's   realization of the joint  measure  ${\bf P}_{X, Y, Z, \widehat{X}}$ induced by the RVs $(X, Y, Z, \widehat{X})$, is  illustrated in 
 Fig.~\ref{fg:onlydecoder}, where  $Z$ is the output of the ``test channel'', ${\bf P}_{Z|X}$. 

Wyner proved the following  coding theorems.

\begin{theorem}  \cite{wyner1978}\\
\label{the_1_dec}
Suppose Assumption~\ref{ass_1} holds. \\
(a) Converse Theorem. For any $\Delta_X\geq 0$, $R^*(\Delta_X) \geq \overline{R}(\Delta_X)$.\\
(b) Direct Theorem. If the conditions stated in \cite[pages~64-65, (i), (ii)]{wyner1978} hold, then  $R^*(\Delta_X) \leq \overline{R}(\Delta_X)$, $0 \leq \Delta_X <\infty$. 
\end{theorem}

%
%

\par When switch $A$ is closed in Fig.~\ref{fg:blockdiagram}, and the tuple of jointly independent and identically distributed RVs $(X^n, Y^n)$, is defined as in Section~\ref{pr:2}, Wyner \cite{wyner1978} generalized Berger's \cite{berger:1971}  characterization of all achievable pairs  $(R,\Delta_X)$, from finite alphabet spaces to abstract alphabet spaces. 

\par A code $(n, M,D_X)$, when switch $A$ is closed in Fig.~\ref{fg:blockdiagram},   is defined as in Section~\ref{pr:2}, with  the encoder $F_E$, replaced by   
\begin{align}
F_{E} : {\cal X}^n \times {\cal Y}^n \longrightarrow {\cal I}_M .
\end{align}
Let ${\cal R}_1$ denote the set of all achievable pairs $(R,\Delta_X)$, again  as defined in Section~\ref{pr:2}. For $\Delta_X\geq 0$, define the infimum of all achievable rates by  
\begin{equation}
\overline{R}_1(\Delta_X) = \inf_{(R,\Delta_X) \in{\cal R}_1}R
\end{equation}

\par Wyner \cite{wyner1978} characterized the infimum of all achievable rates $\overline{R}_1(\Delta_X)$, by the single-letter RDF   $R_{X|Y}(\Delta_X)$ given by (\ref{eq:OP1ck1}), (\ref{eq:OP2ck1}). The  coding  Theorems are given by  Theorem~\ref{the_1_dec} with   $R^*(\Delta_X)$ and $\overline{R}(\Delta_X)$   replaced by $\overline{R}_1(\Delta_X)$  and $R_{X|Y}(\Delta_X)$, respectively. That is, $\overline{R}_1(\Delta_X)=R_{X|Y}(\Delta_X)$ (using Wyner's notation \cite[Appendix~A]{wyner1978}) These coding theorems generalized earlier work of Berger \cite{berger:1971} for finite alphabet spaces. 

Wyner also derived a fundamental lower bound on $R^*(\Delta_X)$ in terms of $\overline{R}_1(\Delta_X)$, as stated in the next remark.

\begin{remark} Wyner \cite[Remarks, page 65]{wyner1978}\\
\label{rem_lb}
(A) For $Z \in {\cal M}(\Delta_X)$, $\widehat{X}=f(Y,Z)$, and thus ${\bf P}_{Z|X,Y}={\bf P}_{Z|X}$,  then by a property of conditional mutual information and the data processing inequality:
\begin{align}
I(X;Z|Y)=I(X;Z, f(Y,Z)|Y) \geq I(X; \widehat{X}|Y) \geq R_{X|Y}(\Delta_X) \label{in_11}
\end{align}
where the last equality is defined since $\widehat{X} \in {\cal M}_0(\Delta_X)$ (see \cite[Remarks, page 65]{wyner1978}). Moreover, 
\begin{align}
{R}^*(\Delta_X) \geq R_{X|Y}(\Delta_X). \label{in_1}
\end{align}
(B) Inequality (\ref{in_1}) holds with equality, i.e., $R^*(\Delta_X) = R_{X|Y}(\Delta_X)$ if the $\widehat{X} \in {\cal M}_0(\Delta_X)$, which achieves $I(X;\widehat{X}|Y)=R_{X|Y}(\Delta_X)$ can be generated as in Fig.~\ref{fg:onlydecoder} with $I(X;Z|Y)= I(X;\widehat{X}|Y)$.  This occurs  if and only if $I(X;Z|\widehat{X},Y)=0$, and follows from  the identity and lower bound  
\begin{align}
I(X;Z|Y)=&I(X;Z,\widehat{X}|Y)=I(X;Z|Y, \widehat{X})+ I(X;\widehat{X}|Y)\\
\geq & I(X;\widehat{X}|Y)
\end{align}
where the  inequality holds with equality   if and only if $I(X;Z|\widehat{X},Y)=0$.
\end{remark}

\subsection{Mean-Square Estimation of Conditionally Gaussian RVs}
Below, state a well-known property of conditionally Gaussian RVs, which we use in our derivations.

\begin{proposition} Conditionally Gaussian RVs\\
\label{prop_cg} 
Consider a pair of multivariate RVs $X=(X_1, \ldots, X_{n_x})\T: \Omega \rightarrow \mathbb{R}^{n_x}$ and $Y=(Y_1, \ldots, Y_{n_y})\T: \Omega \rightarrow \mathbb{R}^{n_y}$, $(n_x,n_y) \in {\mathbb Z}_+ \times {\mathbb Z}_+$,   defined on some probability distribution $\Big(\Omega, {\cal F}, {\mathbb P}\Big)$. Let ${\cal G}\subseteq {\cal F}$ be a sub$-\sigma-$algebra. Assume the conditional distribution of $(X, Y)$ conditioned on ${\cal G}$, i.e.,  ${\bf P}(dx, dy |{\cal G})$ is ${\mathbb P}-$a.s. (almost surely) Gaussian, with conditional means 
\begin{align}
\mu_{X|{\cal G}}\tri{\bf E}\Big(X\Big|{\cal G}\Big), \hst  \mu_{Y|{\cal G}}\tri {\bf E}\Big(Y\Big|{\cal G}\Big)
\end{align}
and  conditional covariances
\begin{align}
Q_{X|{\cal G}} \tri\Cov\Big(X,X\Big|{\cal G}\Big), \hst Q_{Y|{\cal G}} \tri\Cov\Big(Y,Y\Big|{\cal G}\Big),
\end{align} 
\begin{align}
 Q_{X,Y|{\cal G}} \tri\Cov\Big(X,Y\Big|{\cal G}\Big).
\end{align} 
Then, the vectors of conditional expectations $\mu_{X|Y,{\cal G}}\tri{\bf E}\Big(X\Big|Y,{\cal G}\Big)$  and matrices of conditional covariances $Q_{X|Y, {\cal G}}\tri\Cov\Big(X, X\Big|Y,{\cal G}\Big)$ are given, ${\mathbb P}-$a.s., by the following expressions\footnote{If the inverse $Q_{Y|{\cal G}}^{-1}$ does not exists then it is replaced by the pseudo inverse $Q_{Y|{\cal G}}^\dagger$.}:
\begin{align}
&\mu_{X|Y,{\cal G}}=\mu_{X|{\cal G}} + Q_{X,Y|{\cal G}}Q_{Y|{\cal G}}^{-1}\Big( Y- \mu_{Y|{\cal G}}\Big)\label{eq:mean22}, \\ \label{eq:mean}
&Q_{X|Y, {\cal G}}\tri Q_{X|{\cal G}}-Q_{X,Y|{\cal G}}Q_{Y|{\cal G}}^{-1}Q_{X,Y|{\cal G}}\T .
\end{align}
If ${\cal G}$ is the trivial information, i.e., ${\cal G}=\{\Omega, \emptyset\}$, then ${\cal G}$ is removed from the above expressions.
\end{proposition}

Note that ${\cal G}=\{\Omega, \emptyset\}$ then (\ref{eq:mean22}), (\ref{eq:mean}) reduce to the well-known conditional mean and conditional covariance of $X$ conditioned on $Y$.

\section{Proofs of Theorem~\ref{thm_rw} and  Theorem~\ref{thm:dec}} \label{mainres}
\label{sect:proofs}
\par In this section we derive the  the statements of Theorem~\ref{thm_rw} and  Theorem~\ref{thm:dec}. The proofs are based on several  intermediate results, some of which also hold for general abstract alphabet spaces.

\subsection{Side Information at Encoder and  Decoder}

\par We start our analysis with the following achievable lower bound on the conditional mutual information $I(X;\widehat{X}|Y)$, which appears in the definition of $R_{X|Y}(\Delta_X)$, given by   (\ref{eq:OP1ck1}), that  strengthen  Gray's lower bound (\ref{gray_lb_1}), given in  \cite[Theorem~3.1]{gray1973}.

\begin{lemma}Achievable lower bound on conditional mutual information\\
 \label{lem:proof1}
Let $(X, Y, \widehat{X})$ be a triple of arbitrary RVs on the abstract spaces  ${\cal X} \times {\cal Y}\times \widehat{\cal X}$,  with distribution ${\bf P}_{X,Y, \widehat{X}}$ and joint marginal the fixed distribution ${\bf P}_{X,Y}$ of $(X, Y)$. \\
Then the following hold.\\
(a) The inequality holds: 
\begin{align}
I(X;\widehat{X}|Y) \geq I(X;\widehat{X}) - I(X;Y) \label{ineq_1}
\end{align}
Moreover,  if
 \begin{align}
{\bf P}_{X|\widehat{X},Y}={\bf P}_{X|\widehat{X}}- a.s.  \label{mc_1}
\end{align}
or equivalently  $Y \leftrightarrow \widehat{X} \leftrightarrow X$ is a MC then the equality holds, 
\begin{align}
I(X;\widehat{X}|Y) =I(X;\widehat{X}) - I(X;Y) \label{ineq_1_neq}
\end{align}

(b)  If $Y \leftrightarrow \widehat{X} \leftrightarrow X$ is a Markov  chain then the  equality holds  
\begin{align}
R_{X|Y}(\Delta_X)= R_X(\Delta_X) - I(X;Y), \hso \Delta_X \leq {\cal D}_C(X|Y) \label{ineq_1_neq_G}
\end{align}
for a strictly positive set ${\cal D}_C(X|Y)$.
\end{lemma}
\begin{proof} See 
Appendix~\ref{app_A}.  
\end{proof}

The next theorem is   used to derive the characterization of $R_{X|Y}(\Delta_X)$.

\begin{theorem} Achievable lower bound on conditional  mutual information and  mean-square error estimation\\
\label{them_lb}
(a)  Let $(X, Y, \widehat{X})$ be a triple of arbitrary RVs on the abstract spaces  ${\cal X} \times {\cal Y}\times \widehat{\cal X}$,  with distribution ${\bf P}_{X,Y, \widehat{X}}$ and joint marginal the fixed distribution ${\bf P}_{X,Y}$ of $(X, Y)$.\\
Define the conditional mean of $X$ conditioned on $(\widehat{X},Y)$ by 
\begin{align}
\overline{X}^{cm} \tri {\bf E}\Big(X\Big|\widehat{X},Y\Big). 
\end{align}
Then the inequality holds:
\begin{equation}
I(X;\widehat{X}|Y) \geq I(X;\overline{X}^{cm}|Y)  \label{eq:LB}
\end{equation}
Moreover, 
\begin{equation}
 \mbox{if} \hst  \overline{X}^{cm}=\widehat{X}-a.s \hst  \mbox{then} \hst I(X;\widehat{X}|Y) = I(X;\overline{X}^{cm}|Y).  \label{as_eq}
\end{equation}
(b) In part (a) let $(X, Y, \widehat{X})$ be a triple of arbitrary RVs on ${\cal X} \times {\cal Y}\times \widehat{\cal X}={\mathbb R}^{n_x} \times {\mathbb R}^{n_y}\times {\mathbb R}^{n_x}$, $(n_x,n_y) \in {\mathbb Z}_+ \times {\mathbb Z}_+$. \\
 For all measurable functions $(y, \widehat{x})\longmapsto g(y, \widehat{x})\in {\mathbb R}^{n_x}$ the mean-square error satisfies
\begin{align}
{\bf E}\Big\{||X-&g(Y, \widehat{X})||_{{\mathbb R}^{n_x}}^2\Big\} \nonumber \\
&\geq {\bf E}\Big\{||X-{\bf E}\Big(X\Big|Y, \widehat{X}\Big)||_{{\mathbb R}^{n_x}}^2\Big\}, \hso \forall g(\cdot). \label{mse_1}
\end{align}
\end{theorem}
\begin{proof} (\textit{a}) By properties of conditional mutual information \cite{pinsker:1964} then 
\begin{align}
I(X;\widehat{X}| Y)\overset{(1)}=& I(X;\widehat{X},\overline{X}^{cm}| Y) \\ 
\overset{(2)}=& I(X;\widehat{X}| \overline{X}^{cm}, Y)+ I(X;\overline{X}^{cm}| Y)   \\
\overset{(3)}\geq&  I(X;\overline{X}^{cm}|Y) 
\end{align}
where \((1)\) is due to $\overline{X}^{cm}$ is a function of $(Y,\widehat{X})$, and a well-known property of  the mutual information \cite{pinsker:1964}, \((2)\) is due to the chain rule of mutual information  \cite{pinsker:1964}, and \((3)\)   is due to $I(X;\widehat{X}| \overline{X}^{cm}, Y)\geq 0$. If $\widehat{X} = \overline{X}^{cm}$- a.s, then $I(X;\widehat{X}| \overline{X}^{cm}, Y)=0$, and hence the inequality (\ref{eq:LB}) becomes an equality. \\
(\textit{b}) The inequality (\ref{mse_1}) is well-known, due to the orthogonal projection theorem.
\end{proof}

\ \

For jointly Gaussian RVs $(X, Y, \widehat{X})$, in the next theorem  we identify simple sufficient conditions for the lower bound of Theorem~\ref{them_lb} to be achievable. 


\begin{theorem} Sufficient conditions for the lower bounds of Theorem~\ref{them_lb} to be achievable\\
\label{them:lb_g}
Consider the statement of Theorem~\ref{them_lb} for  a  triple of jointly Gaussian RVs  $(X, Y, \widehat{X})$  on ${\mathbb R}^{n_x} \times {\mathbb R}^{n_y}\times {\mathbb R}^{n_x}$, $(n_x,n_y) \in {\mathbb Z}_+ \times {\mathbb Z}_+$, i.e.,  ${\bf P}_{X,Y, \widehat{X}}={\bf P}_{X,Y, \widehat{X}}^G$ and joint marginal the fixed Gaussian distribution ${\bf P}_{X,Y}={\bf P}_{X,Y}^G$ of $(X, Y)$\\
Suppose Conditions 1 and 2  hold:
\begin{align}\label{eq:condA}
&\mbox{Condition 1.}\hso 
{\bf E}\Big(X\Big|Y\Big) ={\bf E}\Big(\widehat{X}\Big|Y\Big)\\
\label{eq:condB}
&\mbox{Condition 2.}\hso \Cov(X,\widehat{X}|Y) \Cov(\widehat{X},\widehat{X}|Y)^{-1} = I_{n_x} 
\end{align}
Then 
\bea
\overline{X}^{cm}=\widehat{X}-a.s
\eea 
 and the inequality (\ref{eq:LB}) holds with equality, i.e.,  $I(X;\widehat{X}|Y) = I(X;\overline{X}^{cm}|Y)$. 
\end{theorem}
\begin{proof} By use of Proposition~\ref{prop_cg}, (\ref{eq:mean22}), and letting $Y = \widehat{X}$ and ${\cal G}$ the information generated by $Y$, then 
\begin{align}
\overline{X}^{cm} \tri &  {\bf E}\Big(X\Big|\widehat{X},Y\Big)\\
=& {\bf E}\Big(X\Big|Y\Big) \\
&+\Cov(X,\widehat{X}|Y) \Cov(\widehat{X},\widehat{X}|Y)^{-1}\Big(\widehat{X} - {\bf E}\Big(\widehat{X}\Big|Y\Big)\Big) \\
\overset{(a)}=& \widehat{X} - a.s. 
\end{align}
where $(a)$ is due to   Conditions 1 and 2.
\end{proof}

\ \

\par Now, we turn our attention to the optimization problem $R_{X|Y}(\Delta_X)$ defined by  (\ref{eq:OP1ck1}), for the multivariate Gaussian source  with mean-square error distortion defined by (\ref{prob_1})-(\ref{prob_9}).  In the next lemma we derive a {\it preliminary parametrization} of the optimal reproduction distribution ${\bf P}_{\widehat{X}|X, Y}$ of  the RDF ${R}_{X|Y}(\Delta_X)$.  \\

\begin{lemma} Preliminary parametrization of  optimal reproduction distribution of $R_{X|Y}(\Delta_X)$\\
\label{lemma:par}
Consider the RDF $R_{X|Y}(\Delta_X)$ defined by (\ref{eq:OP1ck1}) for the multivariate Gaussian source, i.e.,  ${\bf P}_{X,Y}={\bf P}_{X,Y}^G$,  with mean-square error distortion defined by (\ref{prob_1})-(\ref{prob_9}). 

(a) For  every joint distribution ${\bf P}_{X, Y, \widehat{X}}$ there exists a jointly Gaussian distribution denoted by  ${\bf P}_{X,Y, \widehat{X}}^G$, with marginal the fixed distribution ${\bf P}_{X,Y}^G$ , which minimizes $I(X; \widehat{X}|Y)$ and satisfies the average distortion constraint, i.e., with $d_X(x,\widehat{x})=||x-\widehat{x}||_{{\mathbb R}^{n_x}}^2$.

(b)  The conditional reproduction distribution ${\bf P}_{\widehat{X}|X,Y}={\bf P}_{\widehat{X}|X,Y}^G$ is induced by the  parametric realization of $\widehat{X}$ (in terms of $H, G, Q_W$),
\begin{align}
&\widehat{X} = H X + G Y + W,  \label{eq:real}\\
&H \in \mathbb{R}^{n_x\times n_x}, \hso  G \in \mathbb{R}^{n_x\times n_y}, \\
&W \in N(0, Q_W), \; Q_W \succeq 0, \\
&W \hso \mbox{independent of $(X, V)$}.\label{eq:real_4}
\end{align}
and $\widehat{X}$ is a Gaussian RV. 

(c) $R_{X|Y}(\Delta_X)$ is characterized by the  optimization problem.
\begin{align}
{R}_{X|Y}(\Delta_X) \tri& \inf_{{\cal M}_0^G(\Delta_X)}  I(X; \widehat{X}|Y), \hso \Delta_X \in [0,\infty)  \label{eq:OP1_char_1}
\end{align}
where ${\cal M}_0^G(\Delta_X)$ is specified by the set 
\begin{align}
{\cal M}_0^G(\Delta_X)\tri & \Big\{ \widehat{X}: \Omega \rar \widehat{\cal X} : \hso (\ref{eq:real})-(\ref{eq:real_4}) \; \mbox{hold, and} \hso   {\bf E}\big\{||X-\widehat{X}||_{{\mathbb R}^{n_x}}^2\big\}\leq \Delta_X \Big\}.
\end{align}
(d) If there exists $(H, G, Q_W)$ such that $\overline{X}^{cm}=\widehat{X}-a.s$,  then a further lower bound on ${R}_{X|Y}(\Delta_X)$ is achieved in the  subset  ${\cal M}_0^{G,o}(\Delta_X)\subseteq {\cal M}_0^G(\Delta_X) $ defined  by 
\begin{align}
{\cal M}_0^{G,o}(\Delta_X)\tri&  \Big\{ \widehat{X}: \Omega \rar \widehat{\cal X} : \hso (\ref{eq:real})-(\ref{eq:real_4}) \; \mbox{hold}\hso  \widehat{X} = \overline{X}^{cm} - a.s, \hso   {\bf E}\big\{||X-\widehat{X}||_{{\mathbb R}^{n_x}}^2\big\}\leq \Delta_X\Big\}
\end{align}
and the corresponding characterization of the RDF is 
\begin{align}
{R}_{X|Y}(\Delta_X) \tri& \inf_{{\cal M}_0^{G,o}(\Delta_X)}  I(X; \widehat{X}|Y), \hso \Delta_X \in [0,\infty)  \label{eq:OP1_char_1_new}
\end{align}
\end{lemma}
\begin{proof} (a) This is omitted since it is similar to the classical unconditional RDF $R_X(\Delta_X)$  of a Gaussian message $X \in N(0,Q_X)$. (b) By (a) the conditional distribution ${\bf P}_{\widehat{X}|X,Y}^G$ is such that, its conditional mean is linear in $(X,Y)$,  its conditional covariance is nonrandom, i.e., constant, and for fixed $(X, Y)=(x,y)$,  ${\bf P}_{\widehat{X}|X,Y}^G$ is Gaussian. Such a distribution is induced by the parametric realization (\ref{eq:real})-(\ref{eq:real_4}).  (c) Follows from parts (a) and (b). (d) Follows from Theorem~\ref{them:lb_g} and (\ref{mse_1}), by letting $g(y,\widehat{x})=\widehat{x}$. 
\end{proof}

\ \

In the next theorem we identify the optimal triple $(H,G,Q_W)$ such that   $ \overline{X}^{cm}=\widehat{X}-a.s$, and  thus establish its existence.   We  also  characterize the RDF by ${R}_{X|Y}(\Delta_X) \tri \inf_{{\cal M}_0^{G,o}(\Delta_X)}  I(X; \widehat{X}|Y)$, and construct a realization $\widehat{X}$ that achieves it.

\begin{theorem}Characterization of RDF ${R}_{X|Y}(\Delta_X)$\\
\label{thm:proof2}
Consider the RDF $R_{X|Y}(\Delta_X)$, defined by  (\ref{eq:OP1ck1}) for the multivariate Gaussian source  with mean-square error distortion, defined by (\ref{prob_1})-(\ref{prob_9}). 

The characterization of the RDF $R_{X|Y}(\Delta_X)$ is 
\begin{align}
{R}_{X|Y}(\Delta_X) \tri  & \inf_{{\cal Q}(\Delta_X)} I(X; \widehat{X}|Y) \label{eq:optiProbl}\\
=& \inf_{{\cal Q}(\Delta_X)}   \frac{1}{2}\log\Big\{ \det(Q_{X|Y}\Sigma_{\Delta} ^{-1})\Big\}\label{eq:optiProbl_nn}
\end{align}
where 
\begin{align}
{\cal Q}(\Delta_X)\tri& \bigg\{\Sigma_{\Delta}:  \trace \big( \Sigma_{\Delta} \big)\leq{\Delta_X}  \bigg\},  \\
\Sigma_{\Delta} \tri&  {\bf E} \Big\{ \Big( X - \widehat{X} \Big) \Big(X - \widehat{X} \Big)\T \Big\},\\
Q_{X|Y} =&   Q_X - Q_{X,Y} Q_Y^{-1}Q_{X,Y}\T, \hso   Q_{X|Y}-\Sigma_{\Delta}\succeq 0, \label{eq:optiProbl_n}\\
Q_{X,Y} =& Q_X C\T, \hso Q_Y=C Q_X C\T + D D\T
\end{align}
and the optimal reproduction $\widehat{X}$ which achieves ${R}_{X|Y}(\Delta_X)$ is 
\begin{align}
\widehat{X} =&H X + \Big(I_{n_x}- H\Big)Q_{X,Y}Q_Y^{-1} Y +  W \label{eq:realization} \\
 H \tri&I_{n_x} - \Sigma_{\Delta} Q_{X|Y}^{-1}\succeq 0, \hso G \tri\Big(I_{n_x}-H\Big)Q_{X,Y} Q_Y^{-1}, \label{eq:realization_nn_1}  \\
   Q_W \tri& \Sigma_{\Delta}  H^T =  \Sigma_\Delta -\Sigma_\Delta Q_{X|Y}^{-1} \Sigma_\Delta= H \Sigma_\Delta \succeq 0. \label{eq:realization_nn}
\end{align}
Moreover, the realization (\ref{eq:realization}) satisfies, almost surely,  
\begin{align}
&{\bf P}_{X|\widehat{X},Y}={\bf P}_{X|\widehat{X}}, \\
 &  {\bf E}\Big(X\Big|\widehat{X},Y\Big)= \widehat{X},\\ 
&{\bf E}\Big(X\Big|Y\Big)={\bf E}\Big(\widehat{X}\Big|Y\Big)= Q_{X,Y} Q_Y^{-1} Y,  \label{prop_1}
\\
& \Cov(X,\widehat{X}|Y)= \Cov(\widehat{X},\widehat{X}|Y).
\end{align}

\end{theorem}
\begin{proof} See Appendix~\ref{app_B}. 
\end{proof}

\ \

\begin{remark} Structural properties  of realization of Theorem~\ref{thm:proof2}
\label{rk:1_ed}

\par For the characterization of the RDF $R_{X|Y}(\Delta_X)$ of Theorem~\ref{thm:proof2}, for the tuple of multivariate jointly Gaussian RVs  $(X,Y)$, we can proceed one step further  to show that the optimal  $\widehat{X}$ defined by   (\ref{eq:realization})-(\ref{eq:realization_nn}) in terms of the matrices  $\big\{\Sigma_\Delta, Q_{X|Y}, H, Q_W\big\}$, is such that 
\begin{align}
&\mbox{i)} \hso H=H^T\succeq 0, \\
&\mbox{ii)} \hso  \big\{\Sigma_\Delta, \Sigma_{X|Y}, H, Q_W\big\} \hso \mbox{have  spectral decompositions w.r.t the same unitary matrix $U U\T=I_{n_x}$}.
\end{align}
 We show this in Corollary~\ref{cor:equivalent}.
\end{remark}

To prove the structural property of Remark~\ref{rk:1_ed} we use the next corollary, which is a degenerate case of \cite[Lemma~2]{charalambous-charalambous-kourtellaris-vanschuppen-2020} (i.e., the structural properties of test channel of Gorbunov and Pinsker \cite{gorbunov-pinsker1974} nonanticipatory RDF of Markov sources). 

\begin{corollary}
Structural properties of realization of optimal $\widehat{X}$ of characterization of $R_{X|Y}(\Delta_X)$\\
\label{cor:sp_rep}
Consider the characterization of the RDF $ R_{X|Y}(\Delta_X)$ of Theorem~\ref{thm:proof2}. 

  Suppose $Q_{X|Y} \succeq 0$ and $\Sigma_\Delta\succeq 0$ commute, that is, 
\begin{align}
Q_{X|Y} \Sigma_\Delta =\Sigma_\Delta Q_{X|Y} . \label{suff_c}
\end{align}
Then 
\begin{align}
&\mbox{(1)} \hso  H= I_{n_x}-\Sigma_\Delta Q_{X|Y}^{-1}= H\T, \hso Q_W = 
\Sigma_\Delta H\T = \Sigma_\Delta H =H\Sigma_\Delta= Q_W\T  \label{suff_c1} \\
& \mbox{(2)} \hso \big\{\Sigma_\Delta, Q_{X|Y}, H, Q_W\big\} \hso \mbox{have  spectral} \nonumber \\
& \hst  \mbox{decompositions w.r.t the same unitary matrix $U U\T=I_{n_x}, \;U\T U=I_{n_x} $}.
\end{align}
that is, the following hold.
\begin{align}
&Q_{X|Y} =U \diag{\{\lambda_{1},\dots, \lambda_{n_x}\}}U\T, \hso \lambda_1 \geq \lambda_2 \geq \ldots \geq \lambda_{n_x}, \\
&\Sigma_\Delta =U \diag{\{\delta_{1},\dots, \delta_{n_x}\}}U\T,\hso \delta_1 \geq \delta_2 \geq \ldots \geq \delta_{n_x},  \\
&H =U \diag\{1-\frac{\delta_{1}}{\lambda_1},\dots, 1-\frac{\delta_{n_x}}{\lambda_{n_x}}\}U\T,  \\
&Q_W =U \diag\{\big(1-\frac{\delta_{1}}{\lambda_1}\big)\delta_1,\dots, \big(1-\frac{\delta_{n_x}}{\lambda_{n_x}}\big)\delta_{n_x}\}U\T. \label{deco_1}
\end{align}
\end{corollary}
\begin{proof} See Appendix~\ref{app_C}.
\end{proof}

In the next corollary we re-express the realization of $\widehat{X}$ which characterizes the RDF   of Theorem~\ref{thm:proof2}  using a translation of $X$ and $\widehat{X}$, by subtracting their conditional means with respect to $Y$, making use of (\ref{prop_1}). Then we apply  corollary~\ref{cor:sp_rep} to establish    that   the optimal matrices of the RDF $ R_{X|Y}(\Delta_X)$ of Theorem~\ref{thm:proof2} are such that   $\big\{\Sigma_\Delta, Q_{X|Y}, H, Q_W\big\}$ have a  spectral decomposition w.r.t the same unitary matrix $U U\T=I_{n_x}$.


\begin{corollary}
 \label{cor:equivalent}
Equivalent characterization of $ R_{X|Y}(\Delta_X)$\\
Consider the characterization of the RDF $ R_{X|Y}(\Delta_X)$ of Theorem~\ref{thm:proof2}. 
Define the translated RVs
\begin{equation}
\mathbf{X} \tri X- {\bf E}\Big\{X\Big|Y\Big\}=   X- Q_{X,Y}Q_Y^{-1}Y,\hst \mathbf{\widehat{X}}  \tri \widehat{X}  - {\bf E}\Big\{\widehat{X}\Big|Y\Big\}=  \widehat{X}- Q_{X,Y}Q_Y^{-1} Y \label{trans_1}
\end{equation}
where the equalities are due to (\ref{prop_1}). Let  
\begin{align}
& Q_{X|Y} =U \diag{\{\lambda_{1},\dots, \lambda_{n_x}\}}U\T,  \hso U U\T =I_{n_x}, U\T U=I_{n_x},  \hso \lambda_1 \geq \lambda_2 \geq \ldots \geq \lambda_{n_x}, \\
& \overline{\mathbf{X}} \tri U\T \mathbf{X} , \hso  \widehat{\overline{\mathbf{X}}} \tri U\T \widehat{\mathbf{X}}.
\end{align}
Then 
\begin{align}
&\mathbf{\widehat{X}}  = H\mathbf{X}   +W,\label{eq:equivalent_m}\\
&I(X;\widehat{X}|Y) 
=  I(\mathbf{X} ;\mathbf{\widehat{X}}  )= I(U\T\mathbf{X} ;U\T\mathbf{\widehat{X}}  ), \label{eq:equivalent}\\
&{\bf E}\big\|X-\widehat{X}\big\|_{{\mathbb R}^{n_x}}^2 
= {\bf E}\big\|\mathbf{X} - \mathbf{\widehat{X}} \big\|_{{\mathbb R}^{n_x}}^2 = {\bf E}\big\|U\T{\mathbf{X}} -U\T\widehat{{\mathbf{X}}}  \big\|_{{\mathbb R}^{n_x}}^2 =\trace \big( \Sigma_{\Delta} \big).\label{eq:equivalent_d}
\end{align}
where $(H, Q_W)$ are given by (\ref{eq:realization_nn_1}) and   (\ref{eq:realization_nn}).\\
Further,  the   characterization of the RDF $R_{X|Y}(\Delta_X)$ (\ref{eq:optiProbl_nn}) satisfies the following equalities and inequality: 
\begin{align}
{R}_{X|Y}(\Delta_X) \tri  & \inf_{{\cal Q}(\Delta_X)} I(X; \widehat{X}|Y) = \inf_{{\cal Q}(\Delta_X)}   \frac{1}{2}\log \max\Big\{1, \det(Q_{X|Y}\Sigma_{\Delta}^{-1})\Big\} \label{equiv_100}   \\
=& \inf_{ {\bf E}\big\|\mathbf{X} - \mathbf{\widehat{X}} \big\|_{{\mathbb R}^{n_x}}^2\leq \Delta_X   } I(\mathbf{X} ;\mathbf{\widehat{X}}  )  \label{equiv_101}\\
=& \inf_{  {\bf E}\big\|U\T\mathbf{X} - U\T\mathbf{\widehat{X}} \big\|_{{\mathbb R}^{n_x}}^2\leq \Delta_X   } I(U\T\mathbf{X} ;U\T\mathbf{\widehat{X}}  )  \label{equiv_1010}\\ 
\geq & \inf_{ {\bf E}\big\|U\T\mathbf{X} - U\T\mathbf{\widehat{X}} \big\|_{{\mathbb R}^{n_x}}^2\leq \Delta_X   } \sum_{t=1}^{n_x} I(\overline{\mathbf{X}}_t ; \widehat{\mathbf{\overline{X}}}_t  )  \label{equiv_1011}
\end{align}
Moreover, the inequality (\ref{equiv_1011}) is achieved if $Q_{X|Y} \succeq 0$ and $\Sigma_\Delta\succeq 0$ commute, that is, if (\ref{suff_c}) holds, and 
\begin{equation}
R_{X|Y}(\Delta_X) =\inf_{ \sum_{i=1}^{n_x} \delta_i \leq  \Delta_X} \frac{1}{2} \sum_{i=1}^{n_x} \log \max\Big\{1, \frac{\lambda_i}{\delta_i}\Big\} \label{fchar_1}
\end{equation}
where 
\begin{align}
\diag\{{\bf E}\Big(U\T{\mathbf{X}} -U\T\widehat{{\mathbf{X}}}\Big) \Big(U\T{\mathbf{X}} -U\T\widehat{{\mathbf{X}}}\Big)\T\}=\diag\{\delta_1,\delta_2, \ldots, \delta_{n_x}\}. \label{fchar_2}
\end{align}
\end{corollary}
\begin{proof} By Theorem~\ref{thm:proof2}, then 
\begin{align}
\widehat{X} =& HX + GY + W\\
 =& HX + \Big(I-H\Big)Q_{X,Y}Q_Y^{-1}Y + W\\
=& H\Big(X - Q_{X,Y} Q_Y^{-1}Y\Big) + Q_{X,Y}Q_Y^{-1}Y + W \\
  \Longrightarrow  & \hso \widehat{X} -  Q_{X,Y} Q_Y^{-1} Y = H\Big(X - Q_{X,Y}Q_Y^{-1}Y\Big) + W\\
 \Longrightarrow & \hso \mathbf{\widehat{X}}  = H\mathbf{X}   +W. \label{mem}
\end{align}
The last equation establishes  (\ref{eq:equivalent_m}). 
By properties of conditional mutual information and the properties of optimal realization $\widehat{X}$ then the following equalities hold.
\begin{align}
I(X;\widehat{X}|Y)= & I(X-Q_{X,Y}Q_Y^{-1}Y;\widehat{X}-Q_{X,Y}Q_Y^{-1}Y|Y)\\
= & I(\mathbf{X} ;\mathbf{\widehat{X}}  |Y), \hst \mbox{  \hso \mbox{by (\ref{trans_1}), i.e.,  (\ref{prop_1})} } \\
 =& H(\mathbf{\widehat{X}}  |Y) - H(\mathbf{\widehat{X}}  |Y,\mathbf{X} )\\
= & H(\mathbf{\widehat{X}}) - H(\mathbf{\widehat{X}} |Y, \mathbf{X} ), \hso \mbox{by indep. of  $\mathbf{X}$ and $Y$}\\
= & H(\mathbf{\widehat{X}}) - H(\mathbf{\widehat{X}} |\mathbf{X}), \hso \mbox{by indep. of  $W$ and $Y$ for fixed $X$}\\
= & I(\mathbf{X} ;\mathbf{\widehat{X}}  )\\
=&I(U\T\mathbf{X} ; U\T\mathbf{\widehat{X}}  ) \label{eq:equivalent}\\
=&I(\overline{\mathbf{X}}_1, \overline{\mathbf{X}}_2, \ldots, \overline{\mathbf{X}}_{n_x} ; \widehat{\overline{\mathbf{X}}}_1, \widehat{\overline{\mathbf{X}}}_2, \ldots, \widehat{\overline{\mathbf{X}}}_{n_x}  )\\
\geq &\sum_{t=1}^{n_x} I(\overline{\mathbf{X}}_t; \widehat{\overline{\mathbf{X}}}_t  ), \hso \mbox{by mutual independence of $\overline{\mathbf{X}}_t,   t=1,2, \ldots, n_x$} \label{eq:equivalent}   
\end{align}
Moreover,  inequality (\ref{eq:equivalent}) holds with equality if $(\overline{\mathbf{X}}_t; \widehat{\overline{\mathbf{X}}}_t), t=1,2, \ldots, n_x$ are jointly independent. \\
The average distortion function is then given by
\begin{align}
 {\bf E}\big\|X-\widehat{X}\big\|_{{\mathbb R}^{n_x}}^2 &= {\bf E}\big\|X-\widehat{X}- Q_{X,Y} Q_Y^{-1}Y+Q_{X,Y}Q_Y^{-1}Y\big\|_{{\mathbb R}^{n_x}}^2 \\
&={\bf E}\big\|\mathbf{X} - \mathbf{\widehat{X}} \big\|_{{\mathbb R}^{n_x}}^2,   \hso \mbox{by (\ref{trans_1}), i.e.,  (\ref{prop_1})} \\
&={\bf E}\big\|U\T\mathbf{X} -U\T \mathbf{\widehat{X}} \big\|_{{\mathbb R}^{n_x}}^2 = \trace \big( \Sigma_{\Delta}\big), \hso \mbox{by $UU\T=I_{n_x}$}. 
\end{align}
By Corollary~\ref{cor:sp_rep}, if  (\ref{suff_c}) holds, that is,  $Q_{X|Y} \succeq 0$ and $\Sigma_\Delta\succeq 0$ satisfy  $Q_{X|Y} \Sigma_\Delta =\Sigma_\Delta Q_{X|Y}$ (i.e., commute), then (\ref{suff_c1})-(\ref{deco_1}) hold,   and by (\ref{fchar_1}),  then 
\begin{align}
\widehat{\overline{\mathbf{X}}} \tri&  U\T \widehat{\mathbf{X}}=U\T H {\mathbf{X}}+U\T W=U\T \widehat{\mathbf{X}}=U\T H U U\T {\mathbf{X}}+U\T W \\
=& U\T H U  \overline{\mathbf{X}}+U\T W, \hso \mbox{$U\T HU$ is diagonal and $U\T W$ has indep. components}.
\end{align}
Hence,  if  (\ref{suff_c}) holds then the lower bound in (\ref{eq:equivalent}) holds with equality, because  $(\overline{\mathbf{X}}_t; \widehat{\overline{\mathbf{X}}}_t), t=1,2 \ldots, n_x$ are jointly independent. Moreover,  if  (\ref{suff_c}) holds then from, say, (\ref{equiv_100}), the expressions  (\ref{fchar_1}), (\ref{fchar_2}) are obtained. The above equations establish all claims.   
\end{proof}

\ \

\begin{proposition} Theorem~\ref{thm_rw} is correct.
\end{proposition}
\begin{proof}
By invoking 
Corollary~\ref{cor:equivalent}, Theorem~\ref{thm:proof2}  and the convexity of $R_{X|Y}(\Delta_X)$ given by  (\ref{fchar_1}),  then we  arrive at the statements of Theorem~\ref{thm_rw}, which completely characterize the RDF $R_{X|Y}(\Delta_X)$, and constructs a  realization of the  optimal $\widehat{X}$ that achieves it. 
\end{proof}

\ \

Next, we discuss the degenerate case, when the statements of Theorem~\ref{thm:proof2} and Theorem~\ref{thm_rw} reduce to the  RDF ${R}_{X}(\Delta_X)$ of a Gaussian RV $X$ with square-error distortion function. We  illustrate that, the identified structural property of the realization matrices $\big\{\Sigma_\Delta, Q_{X|Y}, H, Q_W\big\}$ leads to  to the well-known water-filling solution.

\begin{remark} Degenerate case of Theorem~\ref{thm:proof2}\\
\label{rem_crdf_1}
Consider the characterization of the RDF ${R}_{X|Y}(\Delta_X)$ of Theorem~\ref{thm:proof2}, and assume $X$ and $Y$ are independent or  $Y$ generates the trivial information, i.e., the $\sigma-$algebra of $Y$, is $\sigma\{Y\}=\{\Omega, \emptyset\}$ or $C=0$ in (\ref{eq:sideInfo})-(\ref{prob_9}).

(a) By the definitions of $Q_{X,Y}, Q_{X|Y}$ then   
\begin{align}
Q_{X,Y}=0, \hso Q_{X|Y}= Q_X. \label{deg_1}
\end{align}
Substituting (\ref{deg_1}) into the expressions of Theorem~\ref{thm:proof2}, then 
 the RDF $R_{X|Y}(\Delta_X)$ reduces to 
\begin{align}
{R}_{X|Y}(\Delta_X)=& {R}_{X}(\Delta_X) \tri \inf_{{\cal Q}(\Delta_X)} I(X; \widehat{X}) \label{eq:optiProbl_deg}\\
=& \inf_{{\cal Q}^{m}(\Delta_X)}   \frac{1}{2}\log \Big\{\det( Q_{X}\Sigma_{\Delta}^{-1}) \Big\}
\end{align}
where 
\begin{align}
{\cal Q}^m(\Delta_X)\tri& \bigg\{\Sigma_{\Delta}:  \trace \big( \Sigma_{\Delta} \big)\leq{\Delta_X}  \bigg\}
\end{align}
and the optimal reproduction $\widehat{X}$ reduces to 
\begin{align}
\widehat{X} =& \Big(I_{n_x} - \Sigma_{\Delta} Q_X^{-1}\Big) X + W, \hso Q_X \succeq \Sigma_\Delta,
 \label{eq:realizatio_x}\\
 Q_W=&\Big( I_{{n_x}} - \Sigma_{\Delta} Q_X^{-1}\Big) \Sigma_{\Delta} \succeq 0. \label{eq:realization_n_deg}
\end{align}
Thus, $R_X(\Delta_X)$ is the well-known RDF of a multivariate memoryless Gaussian RV $X$ with square-error distortion.

(b) For the RDF ${R}_{X}(\Delta_X)$  of part (a),  it is  known \cite{Ihara:1993} that $\Sigma_\Delta$ and $Q_X$ have a spectral decomposition with respect to the same unitary matrix, that is, 
\begin{align}
& Q_X = U\Lambda_X U\T, \hso  \Sigma_\Delta = U\Delta U\T, \hso  UU\T = I \\
& \Lambda_X = \diag{\{ \lambda_{X,1},\dots, \lambda_{X,n_x}\}} , \hso   \Delta = \diag{\{\delta_{1},\dots, \delta_{n_x}\}}
\end{align}
where the entries of $(\Lambda_X, \Delta)$ are in decreasing order. \\
Define
\begin{align}
\mathsf{X}^p \tri  U\T X, \hso  \mathsf{\widehat{X}}^p \tri U\T \widehat{X}, \hso \mathsf{W}^p \tri  U\T W.
\end{align}
Then a parallel channel realization of the optimal reproduction $\mathsf{\widehat{X}}^p$ is obtained given by,
\begin{align}
&\mathsf{\widehat{X}}^p = \mathsf{H}\mathsf{X}^p + \mathsf{W}^p , \\
& \mathsf{H} = I_{n_x} -   \Delta\Lambda_X^{-1}=\diag{\{1-\frac{\delta_1}{\lambda_{X,1}},\dots,1-\frac{\delta_{n_x}}{\lambda_{X,n_x}}\}} , \\
& Q_{\mathsf{W}^p} = \mathsf{H} \Delta = \diag{\{\big(1-\frac{\delta_1}{\lambda_{X,1}}\big)\delta_1,\dots,\big(1-\frac{\delta_{n_x}}{\lambda_{X,n_x}}\big)\delta_{n_x}\}} . 
\end{align}
The RDF $R_{X}(\Delta_X)$ is then computed from the  reverse water-filling equations, as follows.
\begin{equation}
R_{X}(\Delta_X) =\frac{1}{2} \sum_{i=1}^{n_x} \log \frac{\lambda_{X,i}}{\delta_i}
\end{equation}
where
\begin{equation}
\sum_{i=1}^{n_x} \delta_i = \Delta_X, \hst \delta_i= \left\{ \begin{array}{lll} \mu, & \mbox{if} & \mu < \lambda_{X,i} \\ \sigma_i, & \mbox{if} & \mu \geq \lambda_{X,i}  \end{array} \right.
\end{equation}
and where $\mu \in [0,\infty)$ is a Lagrange multiplier (obtained from the Kuch-Tucker conditions).
\end{remark}

\subsection{Side Information only at Decoder}
In general, when the side information is available only at the decoder the achievable  operational rate $R^*(\Delta_X)$ is  greater than the achievable operational rate $\overline{R}_1(\Delta_X)$,  when the side information is available to the encoder and the decoder \cite{wyner1978}. By Remark~\ref{rem_lb},  $\overline{R}(\Delta_X) \geq R_{X|Y}(\Delta_X)$, and equality holds if $I(X;Z|\widehat{X}, Y)=0$.  
 
 \par In view of the characterization of $R_{X|Y}(\Delta_X)$
and the realization of the optimal reproduction $\widehat{X}$ of Theorem~\ref{thm_rw}, which is presented in Fig.~\ref{fg:realization}, we  observe that we can re-write (\ref{eq:realization_sp}) as follows.
\begin{align}
\widehat{X} =& \Big(I_{n_x} - \Sigma_{\Delta} Q_{X|Y}^{-1}\Big) X+ \Sigma_\Delta Q_{X|Y}^{-1} Q_{X,Y} Q_Y^{-1} Y + W, 
 \label{eq:realization-d}\\
 =&\Sigma_\Delta Q_{X|Y}^{-1} Q_{X,Y}Q_Y^{-1}Y +Z \\
 =&f(Y,Z)\\
 Z=&\Big(I_{n_x} - \Sigma_{\Delta} Q_{X|Y}^{-1}\Big)  \Big(X +  \Big(I_{n_x} - \Sigma_{\Delta} Q_{X|Y}^{-1}\Big)^{-1}    W\Big), \\
  H=&I_{n_x}-\Sigma_\Delta Q_{X|Y}^{-1}, \; Q_W=H\Sigma_\Delta,  \hso \mbox{defined by  (\ref{eq:realization_nn_1_sp})-(\ref{spe_d}),}     \label{eq:realization_n}\\
  {\bf P}_{Z|X,Y}=&{\bf P}_{Z|X}, \hso \mbox{$(\widehat{X}, Y)$  uniquely defined $Z$, which implies $I(X;Z|\widehat{X}, Y)=0$.}
\end{align}
The realization $\widehat{X}=f(Y,Z)$ is shown in Fig~\ref{fg:realization}. \\

\begin{proposition} Theorem~\ref{thm:dec} is correct.
\end{proposition}
\begin{proof} From  the above realization of $\widehat{X}=f(Y,Z)$, we have the following. (a) By Wyner, see  Remark~\ref{rem_lb}, then the inequalities (\ref{in_11}) and (\ref{in_1}) hold, and equalities holds if $I(X;Z| \widehat{X}, Y)=0$. That is, for any $\widehat{X}= f(Y,Z)$, by properties of conditional mutual information then 
\begin{align}
I(X;Z|Y) &\overset{(\alpha)}= I(X;Z,\widehat{X}|Y)\\
 &\overset{(\beta)}= I(X;Z|\widehat{X},Y) + I(X;\widehat{X}|Y) \\
 &\overset{(\gamma)}\geq I(X;\widehat{X}|Y ) \label{lb_p}
\end{align}
where $(\alpha)$ is due to $\widehat{X}= f(Y,Z)$, $(\beta)$ is due to the chain rule of mutual information, and $(\gamma)$ is due to   $I(X;Z|\widehat{X},Y)\geq 0$. Hence,  (\ref{lower_b_d_1}) is obtained (as as in Wyner \cite{wyner1978} for a tuple of scalar jointly Gaussian RVs). 
(b) Equality holds in (\ref{lb_p}) if there exists an $\widehat{X}= f(Y,Z)$  such that   $I(X;Z|\widehat{X},Y)=0$, and  the average  distortion is satisfied. Taking    $\widehat{X}= f(Y,Z) = (I_{n_x}-H)Q_{X,Y} Q_Y^{-1} Y + Z$, where  $Z= g(X,W)$ is   specified by (\ref{eq:realization-d})-(\ref{eq:realization_n}),  then $I(X;Z|\widehat{X},Y)=0$ and the average distortion is satisfied.  Since the realization (\ref{eq:realization-d})-(\ref{eq:realization_n})
is identical to the realization (\ref{real_d1})-(\ref{real_d2}), then part (b) is also  shown. (c) This follows directly from the optimal realization. 
\end{proof}


\ \

\begin{figure}
     \centering
     \begin{subfigure}[b]{0.4\textwidth}
         \centering
         \includegraphics[width=\textwidth]{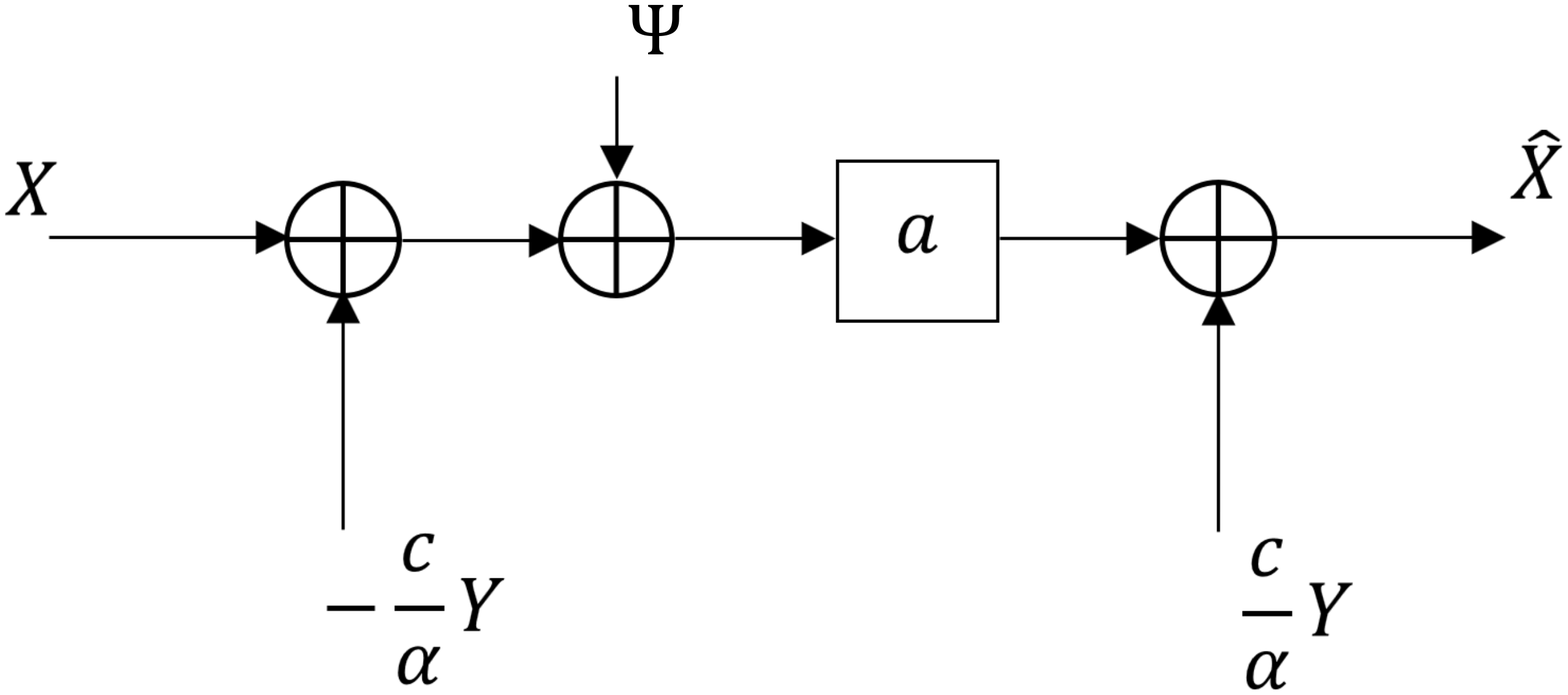}
         \caption{RDF $R_{X|Y}(\Delta_X)$: Wyner's \cite{wyner1978} optimal realization of $\widehat{X}$ for RDF $R_{X|Y}(\Delta_X)$ of (\ref{sc_1})-(\ref{sc_4}).}
          \label{real_ed}
     \end{subfigure}
     \hfill
     \begin{subfigure}[b]{0.4\textwidth}
         \centering
         \includegraphics[width=\textwidth]{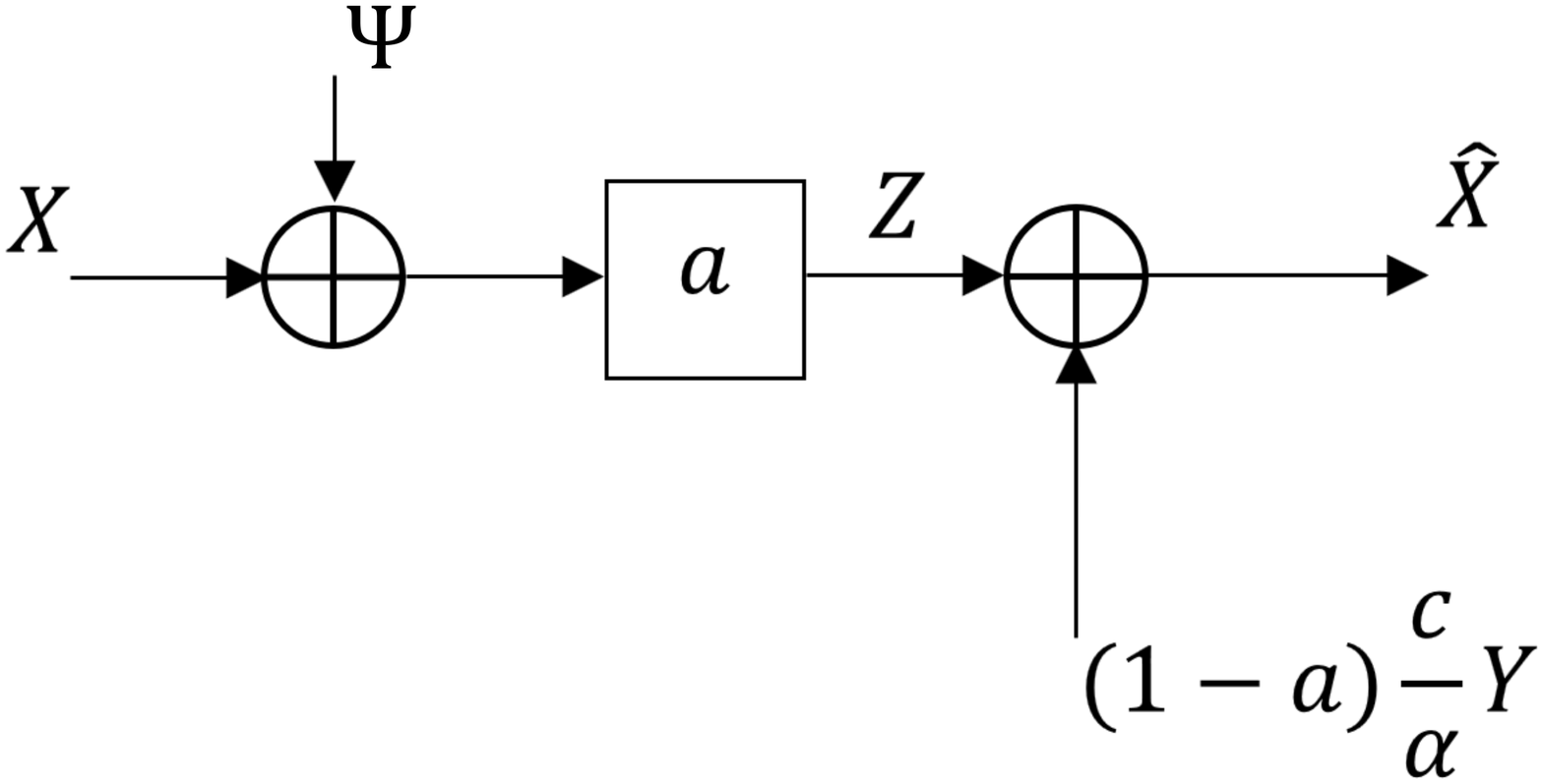}
         \caption{RDF $\overline{R}(\Delta_X)$: Wyner's \cite{wyner1978} optimal  realization   $\widehat{X}=f(X, Z)$ for RDF  $\overline{R}(\Delta_X)$ of (\ref{sc_1})-(\ref{sc_4}).}
         \label{real_d}
     \end{subfigure}
        \caption{Wyner's realizations of optimal reproductions for RDFs $R_{X|Y}(\Delta_X)$ and $\overline{R}(\Delta_X)$}
        \label{fig:three graphs}
\end{figure}

\begin{remark} Relation to Wyner's \cite{wyner1978} optimal test channel realizations\\
\label{rem-wyner}
Now, we verify that our optimal realizations  of  $\widehat{X}$ and closed form expressions for $R_{X|Y}(\Delta_X)$ and $\overline{R}(\Delta_X)$ are identical to  Wyner's \cite{wyner1978} realizations  and RDFs (see  Fig.~\ref{fig:three graphs}),  for  the tuple of scalar-valued, jointly Gaussian RVs $(X, Y)$,  with square error distortion function,  
\begin{align}
&X: \Omega \rar {\cal X} \tri {\mathbb R}, \hso  Y: \Omega \rar {\cal Y} \tri {\mathbb R},\hso  \widehat{X}: \Omega \rar \widehat{\cal X} \tri {\mathbb R},  \label{sc_1}\\
&d_X(x,\widehat{x})=\big(x-\widehat{x}\big)^2,  \label{sc_2} \\
&X \in N(0, \sigma_X^2), \; \sigma_X^2>0, \hso Y=\alpha \Big(X+U\Big), \label{sc_3} \\
& U \in N(0,\sigma_U^2), \hso \sigma_U^2>0, \hso \alpha >0.  \label{sc_4}
\end{align}
(a) RDF $R_{X|Y}(\Delta_X)$. First, we show our realization of  optimal $\widehat{X}$ of   Fig.~\ref{fg:realization} that achieves the RDF of  Theorem~\ref{thm_rw}, degenerates to    Wyner's \cite{wyner1978} optimal realization that achieves the  RDF $R_{X|Y}(\Delta_X)$, for the tuple of scalar-valued, jointly Gaussian RVs $(X, Y)$,  with square error distortion function, given by  (\ref{sc_1})-(\ref{sc_4}). This is verified below. \\ 
By Theorem~\ref{thm_rw}.(a) applied to (\ref{sc_1})-(\ref{sc_4}), then  
\begin{align}
&Q_X=\sigma_X^2, \hso Q_{X,Y}=\alpha \sigma_X^2,\hso Q_{Y}=\sigma_Y^2= \alpha^2 \sigma_X^2 + \alpha^2 \sigma_U^2, \hso Q_{X|Y}= c \sigma_U^2, \hso c\tri \frac{\sigma_X^2}{\sigma_X^2 +\sigma_U^2}, \label{deg_W1} \\
&H=1-\Delta_X Q_{X|Y}^{-1}=\frac{c \sigma_U^2 - d}{c\sigma_U^2} \equiv a, \hso Q_{X,Y}Q_Y^{-1}=\frac{c}{\alpha}, \hso HQ_{X,Y}Q_Y^{-1}=\frac{a c}{\alpha},\label{deg_W2} \\
& W= H\Psi = a \Psi, \hst Q_\Psi= H^{-1} \Delta_X= \frac{\Delta_X}{a}= \frac{c\sigma_U^2\Delta_X}{c\sigma_U^2 -\Delta_X}, \hso  c\sigma_U^2 -\Delta_X>0  . \label{deg_W3}
\end{align}
Moreover, by Theorem~\ref{thm_rw}.(b) the optimal reproduction $\widehat{X}\in {\cal M}_0(d)$ and  $R_{X|Y}(d)$ are, 
\begin{align}
&\widehat{X} = a(X - \frac{c}{\alpha}Y) + \frac{c}{\alpha}Y+ a\Psi,\hso  c\sigma_U^2 -\Delta_X>0 \label{realization_ED_1}\\
\label{eq:rateWynerZivscalar}
&R_{X|Y}(\Delta_X) = \left\{ \begin{array}{ll} \frac{1}{2}\log \frac{c \sigma_U^2}{\Delta_X}, & 0 < \Delta_X <c\sigma_U^2\\
0, &  \Delta_X \geq c\sigma_U^2 . \end{array} \right.
\end{align}
This shows our realization of  Fig.~\ref{fg:realization} degenerates to Wyner's \cite{wyner1978}  realization of Fig.~\ref{real_ed}.

\noi (b) RDF $\overline{R}(\Delta_X)$. Now, we show that our realization of optimal $\widehat{X}=f(Y,Z)$  that achieves the RDF $\overline{R}(\Delta_X)$ of   Theorem \ref{thm:dec}, degenerates  to Wyner's \cite{wyner1978} realization that achieves the  RDF $\overline{R}(\Delta_X)$,  of the tuple of scalar-valued, jointly Gaussian RVs $(X, Y)$,  with square error distortion function given by  (\ref{sc_1})-(\ref{sc_4}).  This is verified below. \\
By Theorem \ref{thm:dec}.(b) applied to  (\ref{sc_1})-(\ref{sc_4}), and using the calculations (\ref{deg_W1})-(\ref{realization_ED_1}),  then 
\begin{align}
&\widehat{X} =f(Y,Z)= \frac{c}{\alpha}(1-a)Y +Z \hst \mbox{by (\ref{realization_ED_1}), (\ref{realization_D_2})}, \label{realization_D_1} \\
 &Z= a \Big(X +\Psi\Big), \hso (a, \Psi) \hst \mbox{defined in (\ref{deg_W2}), (\ref{deg_W3})}     \label{realization_D_2}\\
& \overline{R}(\Delta_X)= R_{X|Y}(\Delta_X)  =\mbox{   (\ref{eq:rateWynerZivscalar}) }    \hst \mbox{by evaluating $I(X;Z)-I(Y;Z)$, i.e.,  using (\ref{rdf_d1_a}) and  (\ref{realization_D_2}).} \label{realization_D_2_n}
\end{align}
This shows our value of $\overline{R}(\Delta_X)$ and optimal realization  $\widehat{X}=f(Y,Z)$,  reproduce Wyner's optimal realization  and value of $\overline{R}(\Delta_X)$ given in  \cite{wyner1978} (i.e.,   Fig.~\ref{real_d}).

\end{remark}

\begin{remark} On   the optimal test channel realization of distributed  source coding problem \cite{tian-chen2009} and \cite{Zahedi-Ostegraard-2014}
 \\
\label{comment}
We  show that contrary to the claim in   \cite[Abstract]{tian-chen2009} and \cite[Theorem~3A]{Zahedi-Ostegraard-2014},  the optimal test channels used in the derivations of the RDF for  the distributed remote source coding problem are incorrect,  and do not produce   Wyner's value of the RDF  $\overline{R}(\Delta_X)$, and  the optimal test channel that achieves it (i.e. the solution  presented in Remark~\ref{rem-wyner}).  
%

(a) Tian and Chen \cite{tian-chen2009} considered the following  formulation\footnote{In the notation of \cite{tian-chen2009} the RVs $(S, X, Y, Z, \widehat{S})$  are represented by $(X, Y, Z,W, \widehat{X})$.}  of (\ref{rdf_po_1}), (\ref{rdf_po_2}):
\begin{align}
\overline{R}^{PO,1}(\Delta_S) \tri\inf_{Z: \: {\bf P}_{Z|X, Y, S}= {\bf P}_{Z|X}, \; \widehat{S}={\bf E}\big\{S\big|Z, Y\big\}, \: {\bf E}\big\{||S-\widehat{S})||^2\big\}\leq \Delta_S} I(X; Z|Y). \label{tian-chen}
\end{align} 
For  multivariate correlated jointly Gaussian RVs $(S,X, Y, Z, \widehat{S})$, with square-error distortion function $d_S(s,\widehat{s})=||s-\widehat{s}||^2$, the RDF  $\overline{R}^{PO,1}(\Delta_S)$ is given  in   \cite[Theorem~4]{tian-chen2009}.\\ Clearly, \\
(i) if  $S=X-$a.s (almost surely) then the  RDF $\overline{R}^{PO,1}(\Delta_S)$ degenerates to Wyner's  RDF $\overline{R}(\Delta_X)$, and  \\
(ii) if   $S=X-$a.s and the RV $X$ is  independent of the RV $Y$ or $Y$ generates a trivial information,   then the  RDF $\overline{R}^{PO,1}(\Delta_S)$ degenerates to the classical RDF of the source $X$, i.e., $R_{X}(\Delta_X)$, as verified  from  (\ref{real_d1})-(\ref{real_d2}), i.e., $Q_{X,Y}=0$ which implies $\widehat{X}=Z$. \\
We examine (i), i.e., under the restriction $S=X-$a.s.,  by recalling  the optimal realization of RVs $(Z, \widehat{X})$ used in the derivation of  \cite[Theorem~4]{tian-chen2009}. \\
The derivation of \cite[Theorem~4]{tian-chen2009}, uses  the following RVs (see \cite[eqn(4)]{tian-chen2009} adopted to our   notation):
\begin{align}
X &= K_{xy}Y + N_1, \\ 
S &= K_{sx}X + K_{sy}Y +N_2 \\
S &= K_{sx}\Big(K_{xy}Y +  N_1\Big) + K_{sy} Y +  N_2,  \\
&= \Big(K_{sx}K_{xy}+K_{sy}\Big) Y + K_{sx}N_1+ N_2
\end{align}
where $N_1$ and $N_2$  are independent Gaussian RVs with zero mean, $N_1$ is independent $Y$ and $N_2$  is independent of $(X, Y)$. \\
The condition  $ X =S-$a.s. implies,
\begin{align}
&K_{sx}K_{xy}+K_{sy} =K_{xy} , \hst K_{sx}=I, \hso N_2=0-a.s.  \label{tian_chen_1}\\
&\Longrightarrow \hso 
K_{xy}+K_{sy} =K_{xy} \hso \Longrightarrow \hso 
K_{sy}=0.
\end{align}
The optimal realization of the auxiliary random variable $Z$ used to achieve the RDF in the derivation of \cite[Theorem~4]{tian-chen2009} (see \cite[3 lines above eqn(32)]{tian-chen2009} using our notation) is 
\begin{align}
 Z&=  UK_{sx}X + N_3\\
&= UX + N_3, \hso \mbox{by (\ref{tian_chen_1}) } \label{tian_chen_2} 
\end{align}
where $U$ is a unitary matrix and $N_3\in N(0, Q_{N_3})$, i.e., Gaussian, such that $Q_{{N_3}}$  is a  diagonal covariance matrix, with  diagonal elements given by,
\begin{align}
\sigma_{3,i}^2 = \frac{\min(\lambda_i,\delta_i)}{\lambda_i-\min(\lambda_i,\delta_i)} \label{tian_chen_3} 
\end{align}
For scalar-valued RVs (\ref{tian_chen_2}), (\ref{tian_chen_3}), reduce to
\begin{align}
Z&= X + N_3, \hst  N_3  \in N \Big(0,\frac{\Delta_X}{\sigma_{X|Y}^2-\Delta_X}\Big)\label{eq:wynerTC}
\end{align}
 It is easy to verify, by letting $(X,Y)$ as in (\ref{sc_3}), (\ref{sc_4}),  that the  realization of the auxiliary RV $Z$ given by (\ref{eq:wynerTC}) is different  from  Wyner's auxiliary RV $Z$ given by (\ref{realization_D_2}), and gives a  value of $I(X;Z)-I(Y;Z)$,  which also different from Wyner's value of the RDF $\overline{R}(\Delta_X)$, i.e.,  (\ref{eq:rateWynerZivscalar}).  In particular, if $\sigma_{X|Y}^2 =\Delta_X$ then it should be $Z=0-$almost surely (as verified from the realization of $Z$ given by  (\ref{realization_D_2}), which reduces to $Z=0-$almost surely, if   $Q_{X|Y}=\Delta_X$, i.e.,  the value of parameter $H$ is  $H=a=0$), but instead the variance of $Z$ takes the value $+\infty$.  \\
We also examine  (ii) (above), i.e., setting   $S=X-$a.s, and taking  $X$ to be independent of $Y$ or $Y$ generates a trivial information. Clearly,    RDFs $\overline{R}^{PO,1}(\Delta_S)$  degenerates to the classical RDF of the source $X$, i.e., $R_{X}(\Delta_X)$, as it is verified  from  (\ref{real_d1})-(\ref{real_d2}), i.e., $Q_{X,Y}=0$ which implies $\widehat{X}=Z$. For scalar-valued RVs the optimal reproduction  $\widehat{X}=Z$ degenerates to (\ref{mem_scal_1}), (\ref{mem_scal_2}). On the other hand,   
(\ref{eq:wynerTC}) does not reduce  to (\ref{mem_scal_1}), (\ref{mem_scal_2}), and moreover the variance of $Z$ defined by (\ref{eq:wynerTC}) is $\sigma_Z^2=\sigma_X^2 + \frac{\Delta_X}{\sigma_{X}^2-\Delta_X}$, and this is fundamentally different from the variance $Q_{\widehat{X}}=\sigma_{\widehat{X}}^2=\sigma_Z^2=\sigma_X^2 -\Delta_X$ of (\ref{mem_scal_2}).  

(b) Similarly to part (a) above,   if we repeat the above steps,  under the condition $S=X-$a.s., the RDF of the remote sensor problem analyzed in  \cite{Zahedi-Ostegraard-2014} reduces to Wyner's RDF $\overline{R}(\Delta_X)$. Moreover,  optimal realization of the auxiliary RV $Z$, which is  used to achieve the RDF in the derivation of \cite[Theorem~3A]{Zahedi-Ostegraard-2014} (see \cite[eqn(26)]{Zahedi-Ostegraard-2014} using our notation)  reduces to  
\begin{align}
Z = U X + \nu   \label{og_1}
\end{align}
where $U$ is a unitary matrix and $\nu\in N(0,Q_{\nu})$ is a zero mean Gaussian vector with  independent components, with  variances across the diagonal of $Q_{\nu}$   given by,
\begin{align}
\sigma_{\nu_i}^2 = \frac{\lambda_i \min{(\lambda_i ,\delta_i)} }{\lambda_i - \min{(\lambda_i ,\delta_i)}}.  \label{og_2}
\end{align}
For scalar-valued RVs (\ref{og_1}), (\ref{og_2}), reduce to 
\begin{align}
Z  = X + \nu, \hso  \nu \in N \bigg(0, \frac{\sigma_{X|Y}^2}{\sigma_{X|Y}^2-\Delta_X}\Delta_X\bigg)\label{og_3}
\end{align}
Clearly, by letting letting $(X,Y)$ as in (\ref{sc_3}), (\ref{sc_4}), the  auxiliary RV $Z$ given by (\ref{og_3}) is different from  Wyner's auxiliary RV $Z$ given by (\ref{realization_D_2}), and does not produce Wyner's value $\overline{R}(\Delta_X)=I(X;Z)-I(Y;Z)$ given by (\ref{realization_D_2_n}).  In particular, if $\sigma_{X|Y}^2 =\Delta_X$ then it should be $Z=0-$almost surely (as verified from the realization of $Z$ given by  (\ref{realization_D_2}), which reduces to $Z=0-$almost surely, if   $Q_{X|Y}=\Delta_X$, i.e.,  the value of parameter $H$ is  $H=a=0$), but instead the variance of $Z$ takes the value $+\infty$. 

It is also noted that the variance of the auxiliary RV $Z$ of \cite{tian-chen2009} given by  (\ref{eq:wynerTC}) is different from the variance of the auxiliary RV of \cite{Zahedi-Ostegraard-2014} given by (\ref{og_3}), although both are designed to achieve the same value of $I(X;Z)-I(Y;Z)$.

\end{remark}

\begin{remark} On the realization of test channels 
\label{rk:1_n}
\par (a)  It should be mentioned that unless a realization of $\widehat{X}$ is identified that achieves the RDFs $R_{X|Y}(\Delta_X)$ and $\overline{R}(\Delta_X)$, such that the joint distribution ${\bf P}_{X,Y, \widehat{X}}$ has marginal the fixed source distribution ${\bf P}_{X,Y}$, then the characterization of the RDFs is incomplete. 

(b) Corollary~\ref{cor:c-dec} follows from the two main theorems, and its  complete solution is generated similarly to the two main theorems.
\end{remark}

\section{Conclusion}
We derived structural properties of  optimal  test channels realizations that achieve the characterizations of RDFs for a tuple of multivariate jointly independent and identically distributed Gaussian random variables  with mean-square error fidelity, when  side information is available to the decoder and not to the encoder, and when side information is available to both. We derived achievable lower bounds on conditional mutual information, and applied   properties of mean-square error estimation to identify  structural properties of optimal test channels that   achieve these bounds. We also applied the structural properties of optimal test channels to  construct realizations of optimal reproductions.  


\section{Appendix}

\subsection{Proof of Lemma~\ref{lem:proof1}}
\label{app_A}
(a) By the chain rule of mutual information then
\begin{align}
I(X;\widehat{X},Y)=&I(X;Y|\widehat{X}) + I(X;\widehat{X})\\
 =& I(X;\widehat{X}|Y) + I(X;Y)
\end{align}
Since $I(X;Y|\widehat{X})\geq 0$ then from above it follows 
 \begin{align}
I(X;\widehat{X}) \leq & I(X;\widehat{X}|Y) + I(X;Y)\\
I(X;\widehat{X}|Y) \geq & I(X;\widehat{X}) - I(X;Y) \label{ineq_1_new}
\end{align}
The above shows (\ref{ineq_1}). 
To show equality, we note the following,
\begin{align*}
I(X;\widehat{X}|Y) &= {\bf E}\Big[\log \frac{{\bf P}_{X|\widehat{X},Y}}{{\bf P}_{X|Y}}\Big]\\
&= {\bf E}\Big[\log \frac{{\bf P}_{X|\widehat{X},Y}}{{\bf P}_{X|Y}}\frac{{\bf P}_{X}}{{\bf P}_{X}}\Big]\\
&= {\bf E}\Big[\log \frac{{\bf P}_{X|\widehat{X},Y}}{{\bf P}_{X}} - \log\frac{{\bf P}_{X|Y}}{{\bf P}_{X}}\Big]\\
&= {\bf E}\Big[\log \frac{{\bf P}_{X|\widehat{X}}}{{\bf P}_{X}} - \log\frac{{\bf P}_{X|Y}}{{\bf P}_{X}}\Big], \hso \mbox{if ${\bf P}_{X|\widehat{X}, Y}={\bf P}_{X|\widehat{X}}$}.
\end{align*}
This completes the statement of equality of  (\ref{ineq_1}), i.e., it establishes equality  (\ref{ineq_1_neq}). (b) Consider a test channel ${\bf P}_{X|\widehat{X},Y}$ such that   ${\bf E}\{||X-\widehat{X}||_{{\mathbb R}^{n_x}}^2\leq \Delta_X$, i.e., $\widehat{X} \in {\cal M}_0(\Delta_X)$,    and such that ${\bf P}_{X|\widehat{X}, Y}={\bf P}_{X|\widehat{X}}$, for $\Delta_X \leq  {\cal D}_C(X|Y)\subseteq [0,\infty)$. By  (\ref{ineq_1_neq})  taking the infimum of both sides over $\widehat{X} \in {\cal M}_0(\Delta_X)$ such that ${\bf P}_{X|\widehat{X}, Y}={\bf P}_{X|\widehat{X}}$  then (\ref{ineq_1_neq_G}) is obtained, for a  nontrivial surface  $\Delta_X \leq  {\cal D}_C(X|Y)$, which exists due to continuity and convexity of $R_{X}(\Delta_X)$ for $\Delta_X \in (0,\infty)$. This completes the proof.


\subsection{Proof of Theorem~\ref{thm:proof2}}
\label{app_B}
We identify the triple $(H,G, Q_W)$  that  satisfied Conditions 1 and 2 of Theorem~\ref{them:lb_g}, which then implies $\widehat{X}=\overline{X}^{mse}$, from which  the claimed statements follow.  
Consider the realization given by  $(\ref{eq:real})$. \\
{\it Condition 1, i.e., $( \ref{eq:condA})$.} The   left hand side part of  $( \ref{eq:condA})$ is  given by (this follows from mean-square estimation theory, or an application of (\ref{eq:mean22}) with  ${\cal G}=\{\Omega, \emptyset\}$)  
\begin{align}
{\bf E}\Big(X\Big|Y\Big)=& {\bf E}\Big(X\Big) + \Cov(X,Y) \Cov(Y,Y)^{-1}\Big(Y - {\bf E}\Big(Y\Big)\Big) \\ \nonumber
 =&\Cov(X,Y) \Cov(Y,Y)^{-1}Y\\
=&Q_{X,Y}Q_Y^{-1} Y\\
 =&Q_XC\T Q_Y^{-1}Y  \hst \mbox{by model (\ref{eq:sideInfo})-(\ref{prob_9})}.        \label{eq:condAL}
\end{align}
Similarly, the right hand side of $( \ref{eq:condA})$ is given by 
\begin{align}
{\bf E}\Big(\widehat{X}\Big|Y\Big) =& {\bf E}\Big(\widehat{X}\Big) + \Cov(\widehat{X},Y) \Cov(Y,Y)^{-1}\Big(Y - {\bf E}\Big(Y\Big)\Big) \\ \nonumber
 =&\Cov(\widehat{X},Y) \Cov(Y,Y)^{-1}Y\\ 
 =&  \Big(HQ_{X,Y} +GQ_Y\Big) Q_Y^{-1}Y  \label{condx_1}  \\     
 =& \Big(HQ_XC\T +GQ_Y\Big) Q_Y^{-1}Y     \hst \mbox{by (\ref{eq:sideInfo})-(\ref{prob_9})} \label{eq:condAR}
\end{align}
Equating  (\ref{eq:condAL}) and (\ref{condx_1}) or  (\ref{eq:condAR}) then 
\begin{align}
&{\bf E}\Big({X}\Big|Y\Big) ={\bf E}\Big(\widehat{X}\Big|Y\Big)\\
&\Longrightarrow \hso Q_{X,Y}Q_Y^{-1}Y=  \Big(HQ_{X,Y} +GQ_Y\Big) Q_Y^{-1}Y   \hst \mbox{by (\ref{condx_1})}  \\
&\Longrightarrow \hso  Q_XC\T Q_Y^{-1}Y=\Big(HQ_XC\T +GQ_Y\Big) Q_Y^{-1}Y  \hst \mbox{by (\ref{eq:sideInfo})-(\ref{prob_9}), (\ref{eq:condAR})}\\
 &\Longrightarrow \hso G=\Big(I-H\Big)Q_XC\T Q_Y^{-1}\\
 &\Longrightarrow \hso G=\Big(I-H\Big)Q_{X,Y} Q_Y^{-1}
\end{align}
Hence, $G$ is obtained, and the reproduction is represented by 
\begin{align}
&\widehat{X}=H X +\Big(I-H\Big) Q_{X,Y} Q_Y^{-1}Y +W, \label{step_1}\\
&\Cov(\widehat{X},Y)=Q_{X,Y}, \hso {\bf E}\Big(\widehat{X}\Big|Y\Big)=Q_{X,Y}Q_Y^{-1}Y={\bf E}\Big(X\Big|Y\Big), \label{step_2}\\
&\widehat{X}- {\bf E}\Big(\widehat{X}\Big|Y\Big)=HX -HQ_{X,Y}Q_Y^{-1} Y +W.\label{step_3}
\end{align}
{\it Condition 2, i.e., $( \ref{eq:condB})$.} 
To apply  $( \ref{eq:condB})$ the following calculations are needed. 
\begin{align}
Q_{X|Y}  &\tri\Cov(X,X|Y)\\
& = {\bf E}  \Big\{ \Big(X - {\bf E}\Big(X\Big|Y\Big)\Big)\Big(X - {\bf E}\Big(X\Big|Y\Big)\Big)\T \Big\} \nonumber \\
&= Q_X - Q_{X,Y} Q_Y^{-1}Q_{X,Y}\T\\ 
&= Q_X - Q_XC\T Q_Y^{-1}CQ_X     \hst \mbox{by (\ref{eq:sideInfo})-(\ref{prob_9})}
\end{align}
\begin{align}
\Cov(X,\widehat{X}|Y) & \tri{\bf E} \Big\{ \Big(X - {\bf E}\Big(X\Big|Y\Big)\Big)\Big(\widehat{X} - {\bf E}\Big(\widehat{X}\Big|Y\Big)\Big)\T \Big\} \nonumber \\
&={\bf E} \Big\{ \Big(X - {\bf E}\Big(X\Big|Y\Big)\Big)\Big(\widehat{X} - {\bf E}\Big(X\Big|Y\Big)\Big)\T \Big\}   \hst \mbox{by (\ref{step_2}) }\\
&={\bf E} \Big\{ \Big(X - {\bf E}\Big(X\Big|Y\Big)\Big)\Big(\widehat{X} \Big)\T \Big\}   \hst \mbox{by orthogonality }\\
&= Q_XH\T - Q_{X,Y}Q_Y^{-1}Q_{Y,X}H\T  \hst \mbox{by (\ref{step_1}), (\ref{step_2}) }\\
&= Q_XH\T - Q_XC\T Q_Y^{-1}CQ_XH\T   \hst \mbox{by  (\ref{eq:sideInfo})-(\ref{prob_9})}    \\ \nonumber
&= \Big(Q_X - Q_XC\T Q_Y^{-1}CQ_X\Big)H\T \\
&=  Q_{X|Y}H\T . \label{eq:condBL}
\end{align}
\begin{align}
\Cov(\widehat{X},\widehat{X}|Y) & \tri{\bf E}  \Big\{ \Big(\widehat{X} - {\bf E}\Big(\widehat{X}\Big|Y\Big)\Big)\Big(\widehat{X} -{\bf E}\Big(\widehat{X}\Big|Y\Big)\Big)\T \Big\}\\ \nonumber
&= HQ_X H\T +Q_W -H Q_{X,Y} Q_Y^{-1}Q_{Y,X}H\T  \hst \mbox{by (\ref{step_3})}\\
&= HQ_XH\T +Q_W - HQ_XC\T Q_Y^{-1}CQ_XH\T   \hst \mbox{by (\ref{eq:sideInfo})-(\ref{prob_9})}  \\ \nonumber
&= H\Big(Q_X-Q_XC\T Q_Y^{-1}CQ_X\Big)H\T + Q_W  \\ 
&=  \label{eq:condBR}H Q_{X|Y}H\T +Q_W . 
\end{align}
By Condition 2 and    (\ref{eq:condBL}) and (\ref{eq:condBR}) then 
\begin{align}
 &\Cov(X,\widehat{X}|Y) \Cov(\widehat{X},\widehat{X}|Y)^{-1} = I_{n_x}\\ \nonumber
&\Longrightarrow \hso Q_{X|Y}H\T\Big(H Q_{X|Y}H\T +Q_W \Big)^{-1} = I_{n_x}\\ 
&\Longrightarrow \hso  Q_W = Q_{X|Y}H\T - H\Sigma_{X|Y}H\T\\ 
&\Longrightarrow\hso  Q_W = \Big(I_{n_x} - H\Big)Q_{X|Y} H\T . \label{eq:KW}
\end{align}
Now, we determine $H$ as follows.
\begin{align}
\Sigma_{\Delta}  \triangleq& \Cov(X,X|Y\widehat{X})\\
 = & \Cov(X,X|Y)-\Cov(X,\widehat{X}|Y)\Cov(\widehat{X},\widehat{X}|Y)^{-1} \Cov(X,\widehat{X}|Y)\T, \ \ \ \ \mbox{by prop.~\ref{prop_cg}, (\ref{eq:mean22})}  \\
 =& \Cov(X,X|Y)- \Cov(X,\widehat{X}|Y)\T, \hso \mbox{by  ( \ref{eq:condB})}\\
 =&Q_{X|Y} - HQ_{X|Y}, \hso \mbox{by (\ref{eq:condBL})}. \label{eq:sigmas} \\
&\Longrightarrow \hso  H Q_{X|Y} = Q_{X|Y} -  \Sigma_{\Delta} \\
& \Longrightarrow \hso H = I-\Sigma_{\Delta} Q_{X|Y}^{-1}  \label{H} 
\end{align}
 Hence, $H$ is obtained. Moreover, $Q_W$ is obtained by  substituting (\ref{H}) into (\ref{eq:KW}). From the above specification of parameters $(H,G,Q_W)$ then the realization (\ref{eq:realization})-(\ref{eq:realization_nn}) follows. From the realization (\ref{eq:realization})-(\ref{eq:realization_nn}) it then follows the property ${\bf P}_{X|\widehat{X},Y}={\bf P}_{X|\widehat{X}}-as$. Moreover, (\ref{eq:optiProbl})-(\ref{eq:optiProbl_n}) are obtained from the realization.

\subsection{Proof of Corollary~\ref{cor:sp_rep}}
\label{app_C}
(a) This part is a special case of a related statement in  \cite{charalambous-charalambous-kourtellaris-vanschuppen-2020}. However, we include it for completeness. By  linear algebra \cite{Horn:2013},  given two matrices $A \in {\cal S}_{+}^{k \times k} , B \in {\cal S}_{+}^{k \times k}$, then the following statements are equivalent: (1) $AB$ is normal, (2) $AB\succeq 0$, where $AB$ normal means $(AB) (AB)\T= (AB)\T (AB)$. Note that $AB$ is normal if and only if  $AB=BA$, i.e., commute. Let $A= U_A D_A U_A\T, B=U_B D_B U_B\T, U_A U_A\T=I_k, U_B U_B\T=I_k$, i.e., there exists a spectral representation of $A, B$  in  terms of unitary matrices $U_A, U_B$ and diagonal matrices $D_A, D_B$. Then, $AB \succeq 0$ if and only if the matrices $A$ and $B$ commute, i.e., $AB =BA$, and $A$ and $B$ commute if and only if  $U_A=U_B$.  \\
Suppose (\ref{suff_c}) holds.  Letting $A=Q_{X|Y}, B=\Sigma_\Delta$, then $A = U_{A} D_{A} U_{A}\T,  
B = U_{B} D_{B} U_{B}\T, U_{A}U_{A}\T=I_{n_x}, U_{B}U_{B}\T=I_{n_x}, U_{A}=U_{B}$. Since  $Q_{X|Y}^{-1}=A^{-1}=U_{A} D_{A}^{-1} U_{A}\T,$  then  $\Sigma_{\Delta} Q_{X|Y}^{-1}= Q_{X|Y}^{-1}\Sigma_{\Delta}$, i.e., they commute. Hence,
\begin{align}
H_t\T=&(I_{n_x}- (\Sigma_{\Delta} Q_{X|Y}^{-1})\T= I_{n_x}- (Q_{X|Y}^{-1})\T \Sigma_{\Delta}\T= I_{n_x}- Q_{X|Y}^{-1} \Sigma_{\Delta} \nonumber \\
=& I_{n_x}-  \Sigma_{\Delta}Q_{X|Y}^{-1}=H   \ \ \ \ \mbox{since $Q_{X|Y}$ and $\Sigma_{\Delta}$ commute.}
\label{com_1}
\end{align}
By the definition of $Q_W$ given by (\ref{eq:realization_nn}) we  have 
\begin{align}
Q_{W} =\Sigma_{\Delta}H\T  =Q_{W}\T= H \Sigma_{\Delta}.
\label{inv_com}
\end{align}
Substituting (\ref{com_1}) into (\ref{inv_com}), then 
\begin{align}
 Q_{W} =\Sigma_{\Delta} H.
\end{align}
Hence, $\{\Sigma_{\Delta}, \Sigma_{X|Y}, H, Q_{W}\}$ are all elements of ${\cal S}_{+}^{p \times p}$ having a spectral decomposition wrt the same unitary matrix $UU\T =I_{n_x}$.
%
%
%
%
%

\section{Acknowledgements}
\par This work was supported in parts by the European Regional Development Fund and the Republic of Cyprus through the Research Promotion Foundation (Project: EXCELLENCE/1216/0365).
\label{Bibliography}
\bibliographystyle{IEEEtran}
\bibliography{references}

\begin{thebibliography}{10}
\providecommand{\url}[1]{#1}
\csname url@samestyle\endcsname
\providecommand{\newblock}{\relax}
\providecommand{\bibinfo}[2]{#2}
\providecommand{\BIBentrySTDinterwordspacing}{\spaceskip=0pt\relax}
\providecommand{\BIBentryALTinterwordstretchfactor}{4}
\providecommand{\BIBentryALTinterwordspacing}{\spaceskip=\fontdimen2\font plus
\BIBentryALTinterwordstretchfactor\fontdimen3\font minus
  \fontdimen4\font\relax}
\providecommand{\BIBforeignlanguage}[2]{{%
\expandafter\ifx\csname l@#1\endcsname\relax
\typeout{** WARNING: IEEEtran.bst: No hyphenation pattern has been}%
\typeout{** loaded for the language `#1'. Using the pattern for}%
\typeout{** the default language instead.}%
\else
\language=\csname l@#1\endcsname
\fi
#2}}
\providecommand{\BIBdecl}{\relax}
\BIBdecl

\bibitem{wyner1978}
A.~Wyner, ``The rate-distortion function for source coding with side
  information at the decoder-ii: General sources,'' \emph{Information and
  Control}, vol.~38, no.~1, pp. 60 -- 80, 1978.

\bibitem{tian-chen2009}
C.~{Tian} and J.~{Chen}, ``Remote vector gaussian source coding with decoder
  side information under mutual information and distortion constraints,''
  \emph{IEEE Transactions on Information Theory}, vol.~55, no.~10, pp.
  4676--4680, 2009.

\bibitem{Zahedi-Ostegraard-2014}
A.~{Zahedi}, J.~{Ostergaard}, S.~H. {Jensen}, P.~{Naylor}, and S.~{Bech},
  ``Distributed remote vector gaussian source coding with covariance distortion
  constraints,'' in \emph{2014 IEEE International Symposium on Information
  Theory}, 2014, pp. 586--590.

\bibitem{wyner-ziv1976}
A.~{Wyner} and J.~{Ziv}, ``The rate-distortion function for source coding with
  side information at the decoder,'' \emph{IEEE Transactions on Information
  Theory}, vol.~22, no.~1, pp. 1--10, January 1976.

\bibitem{berger:1971}
T.~Berger, \emph{Rate Distortion Theory:~A Mathematical Basis for Data
  Compression}.\hskip 1em plus 0.5em minus 0.4em\relax Englewood Cliffs, NJ:
  Prentice-Hall, 1971.

\bibitem{draper-wornell2004}
S.~C. {Draper} and G.~W. {Wornell}, ``Side information aware coding strategies
  for sensor networks,'' \emph{IEEE Journal on Selected Areas in
  Communications}, vol.~22, no.~6, pp. 966--976, 2004.

\bibitem{gray1973}
R.~M. Gray, ``A new class of lower bounds to information rates of stationary
  sources via conditional rate-distortion functions,'' \emph{IEEE Transactions
  on Information Theory}, vol.~19, no.~4, pp. 480--489, July 1973.

\bibitem{Oohama1997}
Y.~{Oohama}, ``{G}aussian multiterminal source coding,'' \emph{IEEE
  Transactions on Information Theory}, vol.~43, no.~6, pp. 1912--1923, Nov
  1997.

\bibitem{Oohama2005}
------, ``Rate-distortion theory for {G}aussian multiterminal source coding
  systems with several side informations at the decoder,'' \emph{IEEE
  Transactions on Information Theory}, vol.~51, no.~7, pp. 2577--2593, July
  2005.

\bibitem{ViswanathanCEO1997}
H.~{Viswanathan} and T.~{Berger}, ``The quadratic {G}aussian ceo problem,''
  \emph{IEEE Transactions on Information Theory}, vol.~43, no.~5, pp.
  1549--1559, Sep. 1997.

\bibitem{SUlukus2012}
E.~{Ekrem} and S.~{Ulukus}, ``An outer bound for the vector {G}aussian ceo
  problem,'' in \emph{2012 IEEE International Symposium on Information Theory
  Proceedings}, July 2012, pp. 576--580.

\bibitem{JunChen2014}
J.~{Wang} and J.~{Chen}, ``Vector {G}aussian multiterminal source coding,''
  \emph{IEEE Transactions on Information Theory}, vol.~60, no.~9, pp.
  5533--5552, Sep. 2014.

\bibitem{Horn:2013}
R.~A. Horn and C.~R. Johnson, Eds., \emph{Matrix Analysis}, 2nd~ed.\hskip 1em
  plus 0.5em minus 0.4em\relax New York, NY, USA: Cambridge University Press,
  2013.

\bibitem{Gallager:1968}
R.~G. Gallager, \emph{Information Theory and Reliable Communication}.\hskip 1em
  plus 0.5em minus 0.4em\relax New York, NY, USA: John Wiley \& Sons, Inc.,
  1968.

\bibitem{pinsker:1964}
M.~S. Pinsker, \emph{The information stability of {G}aussian random variables
  and processes}, 1964, vol. 133, no.~1.

\bibitem{charalambous-charalambous-kourtellaris-vanschuppen-2020}
C.~Charalambous, T.~Charalambous, C.~Kourtellaris, and J.~van Schuppen,
  ``Structural properties of nonanticipatory epsilon entropy of multivariate
  {G}aussian sources,'' in \emph{2020 IEEE International Symposium on
  Information Theory}, 2020.

\bibitem{gorbunov-pinsker1974}
A.~K. Gorbunov and M.~S. Pinsker, ``Prognostic epsilon entropy of a {G}aussian
  message and a {G}aussian source,'' \emph{Problems of Information
  Transmission}, vol.~10, no.~2, pp. 93--109, 1974.

\bibitem{Ihara:1993}
S.~Ihara, \emph{{Information theory for continuous systems}}.\hskip 1em plus
  0.5em minus 0.4em\relax Singapore: World Scientific, 1993.

\end{thebibliography}
\end{document}